\documentclass[a4paper,
]{scrartcl}
\usepackage{url}
\newcommand{\email}[1]{\url{#1}}
\usepackage[utf8x]{inputenc}% WICHTIG! Für rôle etc.
\usepackage{xcolor}
\usepackage{amssymb}
\usepackage{amsmath}
\usepackage{amsthm}
\usepackage{stmaryrd}
\usepackage{textcomp}
\usepackage{pifont}
\usepackage{bbding}
\usepackage{graphicx}
\usepackage[english]{babel}
\usepackage{proof}
\usepackage{turnstile}
\usepackage{tikz}
\usepackage{hyperref}
\usepackage[numbers %numern als referenzen
]{natbib}   % better bibliography
\usepackage{arxiv}

\newtheorem{theorem}{Theorem}[section]
\newtheorem{corollary}{Corollary}[section]
\newtheorem{lemma}{Lemma}[section] 

\newtheorem*{claim*}{Claim}

\newtheorem{proposition}{Proposition}[section]
\theoremstyle{definition}      % Definition with upshape body instead of italics   
\newtheorem{definition}{Definition}[section]

\newtheorem{example}{Example}[section]

\deffootnote{1em}{1em}{\textsuperscript{\thefootnotemark\ }}

\title{Amortised Resource Analysis and Typed Polynomial Interpretations\\
(extended version)\thanks{This research is partly supported by FWF (Austrian Science Fund) project P25781.}}
\author{Martin Hofmann\\
Institute of Computer Science,\\
LMU Munich, Germany\\
email: \email{hofmann@ifi.lmu.de}
\and Georg Moser\\
Institute of Computer Science,\\
University of Innsbruck, Austria,\\
email: \email{georg.moser@uibk.ac.at}
}

\begin{document}

\maketitle

\begin{abstract}
We introduce a novel resource analysis for
typed term rewrite systems based on a potential-based
type system. This type system gives rise to polynomial bounds 
on the innermost runtime complexity. We relate the thus obtained amortised resource analysis to 
polynomial interpretations and obtain the perhaps surprising result 
that whenever a rewrite system $\RS$ can be
well-typed, then there exists a polynomial interpretation 
that orients $\RS$. For this we adequately adapt the standard notion
of polynomial interpretations to the typed setting.

\end{abstract}

\section{Introduction}
\label{Introduction}

In recent  years there have  been several approaches to  the automated
analysis of the  complexity of programs. Mostly  these approaches have
been  developed  independently  in  different communities  and  use  a
variety of different,  not easily comparable techniques. 
Without hope for completeness,
we mention work by Albert et al.~\cite{AAGP:2011} that underlies \costa, an automated tool
for the resource analysis of \Java\ programs. 
Related work, targeting \C\ programs, has been reported by 
Alias et al.~\cite{ADFG:2010}. In Zuleger et al.~\cite{ZGSV:2011} 
further approaches for the runtime complexity analysis of \C\ programs is reported, 
incorporated into \loopus. Noschinski et al.~\cite{NoschinskiEG13} 
study runtime complexity analysis of rewrite systems, 
which has been incorporated in \aprove. Finally,
the \raml\ prototype~\cite{HAH:2012} provides an automated potential-based 
resource analysis for various resource bounds of functional programs 
and~\tct~\cite{AvanziniM13a} is one of the most powerful tools 
for complexity analysis of rewrite systems. 

Despite the abundance in the literature almost no comparison results
are known that relate the sophisticated methods developed. Indeed a precise
comparison often proves difficult. For example, on the surface there is an obvious
connection between the decomposition techniques established by Gulwani and
Zuleger in~\cite{GulwaniZuleger:2010} and recent advances on this topic in the complexity
analysis of rewrite systems, cf.~\cite{AvanziniM13}. 
However, when investigated in detail, precise comparison results are difficult to
obtain.
We exemplify the situation with a simple example that will also
serve as running example throughout the paper.

\begin{example}
\label{ex:1}
Consider the following term rewrite system (TRS for short) $\RSa$, 
encoding a variant of an example by Okasaki~\cite[Section~5.2]{Okasaki:1999}.
\begin{alignat*}{4}
1\colon && \checkF(\queue(\nil,r)) &\to \queue(\rev(r),\nil) 
&
\hspace{6ex}
7 \colon && \enq(\zero) &\to \queue(\nil,\nil)
\\
2\colon &\!& \checkF(\queue(x \cons xs,r)) &\to \queue(x \cons xs,r)
&
\hspace{6ex}
8\colon && \revp(\nil,ys) & \to ys
\\
3\colon && \tail(\queue(x \cons f,r)) &\to \checkF(\queue(f,r))
&
\hspace{6ex}
9 \colon && \rev(xs) &\to  \revp(xs,\nil)
\\
4\colon && \snoc(\queue(f,r),x) &\to \checkF(\queue(f,x \cons r))
&
\hspace{6ex}
10\colon &\!& \head(\queue(x \cons f,r)) &\to x
\\
5\colon && \revp(x \cons xs,ys) &\to \revp(xs,x \cons ys)
&
\hspace{6ex}
11\colon && \head(\queue(\nil,r)) &\to \errorHead
\\
6 \colon && \enq(\mS(n)) &\to \snoc(\enq(n),n)
&
\hspace{6ex}
12\colon && \tail(\queue(\nil,r)) &\to \errorTail
\end{alignat*}
$\RSa$ encodes an efficient implementation of a queue in
functional programming. A queue is represented as a pair
of two lists $\queue(f,r)$, encoding the initial part~$f$ 
and the reversal of the remainder~$r$. Invariant of the algorithm is
that the first list never becomes empty, which is achieved by 
reversing $r$ if necessary. Should the invariant ever be violated,
an exception ($\errorHead$ or $\errorTail$) is raised. 
\end{example}

We exemplify the physicist's method of amortised analysis~\cite{Tarjan:1985}. 
We assign to every queue $\queue(f,r)$ the length of $r$ as \emph{potential}. Then the amortised 
cost for each operation is constant, as the costly reversal operation
is only executed if the potential can pay for the operation, 
compare~\cite{Okasaki:1999}. Thus, based on an amortised analysis, we
deduce the optimal linear runtime complexity for $\RS$.

On the other hand let us attempt an application of the interpretation method 
to this example. Termination proofs by interpretations are well-established
and can be traced back to work by Turing~\cite{Turing:49}. We note that $\RSa$
is polynomially terminating. Moreover, it is rather straightforward
to restrict so-called \emph{polynomial interpretations}~\cite{BN98} suitably so that
compatibility of a TRS $\RS$ induces polynomial runtime complexity, 
cf.~\cite{BonfanteCMT01}. Such polynomial interpretations are called \emph{restricted}.
However, it turns out that no restricted polynomial interpretation can exist
that is compatible with $\RSa$. The reasoning is simple. The constraints induced by 
$\RSa$ imply that the function $\snoc$ has to be interpreted by a linear polynomial. 
Thus an exponential interpretation is required for enqueuing ($\enq$). 
Looking more closely at the different proofs, we observe the following. While in
the amortised analysis the potential of a queue $\queue(f,r)$ depends only
on the remainder $r$, the interpretation of $\queue$ has to be monotone in
both arguments by definition. This difference induces that $\snoc$ is
assigned a strongly linear potential in the amortised analysis, 
while only a linear interpretation is possible for $\snoc$. 

Still it is possible to precisely relate
amortised analysis to polynomial interpretations if we base our
investigation on many-sorted (or typed) TRSs and make suitable use of the
concept of \emph{annotated types} originally introduced in~\cite{HofmannJ03}.
More precisely, we establish the following results.
We establish a novel runtime complexity analysis for  
typed constructor rewrite systems. This complexity analysis is 
based on a potential-based amortised analysis incorporated into a type system.
From the annotated type of a term its derivation height with respect to innermost rewriting
can be read off (see Theorem~\ref{l:2}).
The correctness proof of the obtained bound rests
on a suitable big-step semantics for rewrite systems, 
decorated with counters for the derivation height of the evaluated
terms. We complement this big-step semantics with a similar decorated
small-step semantics and prove equivalence between these semantics. Furthermore
we strengthen our first result by a similar soundness result based on
the small-step semantics (see Theorem~\ref{l:5}).
Exploiting the small-step semantics we prove our main result
that from the well-typing of a TRS $\RS$ we can read off a typed
polynomial interpretation that orients $\RS$ (see Theorem~\ref{l:8}). 

While the type system exhibited is inspired by Hoffmann et al.~\cite{HoffmannH10a}
we generalise their use of annotated types to arbitrary (data) types. Furthermore
the introduced small-step semantics (and our main result) directly establish
that any well-typed TRS is terminating, thus circumventing the notion of
partial big-step semantics introduced in~\cite{HoffmannH10b}. Our main result
can be condensed into the following observations. 
The physicist's method of amortised analysis conceptually amounts to the interpretation
method if we allow for the following changes: 
\begin{itemize}
\item Every term bears a potential, not only constructor terms.
\item Polynomial interpretations are defined over annotated types.
\item The standard compatibility constraint is weakened to orientability, that is,
all ground instances of a rule strictly decrease.
\end{itemize}
Our study is purely theoretic, and we have not (yet) attempted an implementation
of the provided techniques. However, automation appears straightforward. Furthermore
we have restricted our study to typed (constructor) TRSs. 
In the conclusion we sketch application of the established results to innermost
runtime complexity analysis of untyped TRSs. 

\medskip
This paper is structured as follows. In the next section we cover some basics
and introduce a big-step operational semantics for typed TRSs. In Section~\ref{AnnotatedTypes}
we clarify our definition of annotated types and provide the mentioned type system. We also
present our first soundness result. In Section~\ref{SmallStep} we 
introduce a small-step operational semantics and prove our second soundness result. 
Our main result will be stated and proved in Section~\ref{TypedPI}. Finally, we conclude
in Section~\ref{Conclusion}, where we also mention future work. 

\section{Typed Term Rewrite Systems}
\label{TypedTRS}

Let $\CS$ denote a finite, non-empty set of \emph{constructor symbols} and
$\DS$ a finite set of \emph{defined function symbols}.
Let $S$ be a finite set of (data) types. A family
$(X_A)_{A \in S}$ of sets is called \emph{$S$-typed} and denotes as $X$. 
Let $\VS$ denote an $S$-typed set of \emph{variables}, such that
the $\VS_s$ are pairwise disjoint. In the following, variables
will be denoted by $x$, $y$, $z$, \dots, possibly extended
by subscripts.

Following~\cite{JR99}, a \emph{type declaration} is of form 
$\typdcl{A_1 \times \cdots \times A_n}{C}$, where $A_i$ and $C$ are types. 
Type declarations serve as input-output specifications for function symbols. 
We write $A$ instead of $\typdcl{}{A}$.
A \emph{signature} $\FS$ (with respect to the set of types $S$) 
is a mapping from $\CS \cup \DS$ to type declarations. 
We often write $\typed{f}{\typdcl{A_1 \times \cdots \times A_n}{C}}$
if $\FS(f) = \typdcl{A_1 \times \cdots \times A_n}{C}$ and
refer to a type \emph{declaration} as a type, if no confusion can arise.
We define the $S$-typed set of terms $\TA(\DS\cup\CS,\VS)$
(or $\TA$ for short): (i) for each $A \in S$: $\VS_A \subseteq \TA_{A}$,
(ii) for $f \in \CS \cup \DS$ such that $\FS(f) = \typdcl{A_1,\dots,A_n}{A}$
and $t_i \in \TA_{A_i}$, we have $f(t_1,\dots,t_n) \in \TA_{A}$.
Type assertions are denoted $\typed{t}{C}$. 
Terms of type $A$ will sometimes be referred to as instances of $A$: 
a term of list type, is simply called a list.
If $t \in \TA(\CS,\varnothing)$ then $t$ is called a \emph{ground constructor term}
or a \emph{value}. The set of values is denoted $\TA(\CS)$.
The ($S$-typed) set of variables of a term $t$ is denoted $\Var(t)$.
The root of $t$ is denoted $\rt(t)$ and the size of $t$, that is
the number of symbols in $t$, is denoted as $\size{t}$. 
In the following, terms are denoted by $s$, $t$, $u$, \dots, possibly extended
by subscripts. Furthermore, we use $v$ (possibly indexed) to denote values.

A \emph{substitution} $\sigma$ is a mapping from variables to terms that respects
types. Substitutions are denoted as sets of assignments: 
$\sigma = \{x_1 \mapsto t_1, \dots, x_n \mapsto t_n\}$.
We write $\dom(\sigma)$ ($\range(\sigma)$) to denote the domain (range) of $\sigma$; 
$\Vrg(\sigma) \defsym \Var(\range(\sigma))$.
Let $\sigma$ be a substitution and $V$ be a set of variables;
$\rst{\sigma}{V}$ denotes the restriction of the domain of $\sigma$ to $V$.
The substitution $\sigma$ is called a \emph{restriction} of a substitution
$\tau$ if $\rst{\tau}{\dom(\sigma)} = \sigma$. Vice versa, $\tau$ is
called \emph{extension} of $\sigma$. Let $\sigma$, $\tau$ be substitutions
such that $\dom(\sigma) \cap \dom(\tau) = \varnothing$. Then we
denote the (disjoint) union of $\sigma$ and $\tau$ as $\sigma \dunion \tau$.
We call a substitution $\sigma$ \emph{normalised} if all terms in the
range of $\sigma$ are values. In the following, all considered substitutions
will be normalised. 

A \emph{typing context} is a mapping from variables $\VS$ to types. 
Type contexts are denoted by upper-case Greek letters.
Let $\Gamma$ be a context and let $t$ be a term. 
The typing relation $\tjudge{\Gamma}{}{}{\typed{t}{A}}$ expresses 
that based on context $\Gamma$, $t$ has type $A$ (with respect to the signature $\FS$). 
The typing rules that define the typing relation are given in Figure~\ref{fig:2}, where we
forget the annotations. 
In the sequel we sometimes make use of an abbreviated notation for sequences of 
types $\vec{A} = A_1,\dots,A_n$ and terms $\vec{t} \defsym t_1,\ldots,t_n$. 

A typed rewrite rule is a pair $l \to r$ of terms, such that 
(i) the type of $l$ and $r$ coincides, 
(ii) $\rt(l) \in \DS$,  and (iii) $\Var(l) \supseteq \Var(r)$.
An $S$-typed \emph{term rewrite system} (\emph{TRS} for short) over the signature
$\FS$ is a finite set of typed rewrite rules. 
We define the \emph{innermost rewrite relation} $\rsrew$ for typed TRSs $\RS$.
For terms $s$ and $t$, $s \rsrew t$ holds, if there exists a context
$C$, a normalised substitution $\sigma$ and a rewrite rule ${l \to r} \in \RS$
such that $s = C[l\sigma]$, $t=C[r\sigma]$ and $s$, $t$ are well-typed.
In the sequel we are only concerned with \emph{innermost} rewriting.
A TRS is \emph{orthogonal} if
it is left-linear and non-overlapping~\cite{BN98,TeReSe}.
A TRS is \emph{completely defined} if all ground normal-forms are values.
These notions naturally extend to typed TRS. In particular, note
that an orthogonal typed TRS is confluent.

\begin{definition}
\label{d:runtimecomplexity}
We define the \emph{runtime complexity} (with respect to $\RS$) as follows:
\begin{equation*}
  \rc(n) \defsym \max \{ \dheight(t,\to) \mid \text{$t$ is basic and 
    $\size{t} \leqslant n$}\} \tkom
\end{equation*}
where a term $t = f(t_1,\dots,t_k)$ is called \emph{basic}
if $f$ is defined, and the terms $t_i$ are only built over 
constructors and variables.
\end{definition}

\begin{figure}[t]
  \centering
  \begin{tabular}{c} 
   $\infer{\eval{\sigma}{0}{x}{v}}{%
     x\sigma = v%
     }$
   \hfill
   $\infer{\eval{\sigma}{0}{c(x_1,\dots,x_n)}{c(v_1,\dots,v_n)}}{%
       c \in \CS
       &
       x_1\sigma = v_1 
       & \cdots
       & x_n\sigma = v_n%
     }$
   \\[3ex]
   $\infer{\eval{\sigma}{m+1}{f(x_1,\ldots,x_n)}{v}}{%
     f(l_1,\ldots,l_n) \to r \in \RS
     & \exists \tau \ \forall i\colon x_i\sigma = l_i\tau
     & \eval{\sigma \dunion \tau}{m}{r}{v}
   }$
   \\[3ex]
   $\infer{\eval{\sigma}{m}{f(t_1,\ldots,t_n)}{v}}{%
     \begin{minipage}[b]{30ex}
       all $x_i$ are fresh \hfil
       \\[1.5ex]
       $\eval{\sigma \dunion \rho}{m_0}{f(x_1,\ldots,x_n)}{v}$ \hfill
     \end{minipage}
     & \eval{\sigma}{m_1}{t_1}{v_1}
     & \cdots
     & \eval{\sigma}{m_n}{t_n}{v_n}
     & m = \sum_{i=0}^n m_i% 
   }$
  \end{tabular}
  
\bigskip
\noindent
  Here $\rho \defsym \{x_1 \mapsto v_1,\ldots,x_n \mapsto v_n\}$.
  Recall that $\sigma$, $\tau$, and $\rho$ are normalised. 
  \caption{Operational Big-Step Semantics}
  \label{fig:1}
\end{figure}

We study \emph{typed} \emph{constructor} TRSs $\RS$, 
that is, for each rule $f(l_1,\dots,l_n) \to r$, the $l_i$ 
are constructor terms. Furthermore, we restrict to \emph{completely defined}
and \emph{orthogonal} systems. These restrictions are natural in the
context of functional programming. 
If no confusion can arise from this, we simply call $\RS$ a TRS.
$\FS$ denotes the signature underlying $\RS$. 
In the sequel, $\RS$  and $\FS$ are kept fixed.

\begin{example}[continued from Example~\ref{ex:1}]
\label{ex:2}
Consider the TRS $\RSa$ and let $S = \{\Nat,\List,\Queue\}$, where
$\Nat$, $\List$, and $\Queue$ represent the type of natural numbers, 
lists over over natural number, and queues respectively. 
Then $\RSa$ is an $S$-typed TRSs over signature $\FS$, where the signature
of some constructors is as follows: 
\begin{alignat*}{4}
  \zero\colon &\Nat & \hspace{6ex}
  \mS\colon & \typdcl{\Nat}{\Nat} & 
  \queue\colon & \typdcl{\List \times \List}{\Queue}
  \\
  \nil\colon & \List & \hspace{6ex}
  \cons\colon & \typdcl{\Nat \times \List}{\List} \tpkt \hspace{5ex}
\end{alignat*}
In order to exemplify the type declaration of defined function symbols, consider
\begin{equation*}
 \typed{\snoc}{\typdcl{\Queue \times \Nat}{\Queue}} \tpkt 
\end{equation*}
\end{example}

As $\RS$ is completely defined any derivation ends in a value. On the other hand, as 
$\RS$ is non-overlapping any innermost derivation is determined modulo the order 
in which parallel redexes are contracted.
This allows us to recast innermost rewriting into an operational
big-step semantics instrumented with resource counters, cf.~Figure~\ref{fig:1}.
The semantics closely resembles similar definitions given in
the literature on amortised resource analysis (see for example~\cite{JLHSH09,HoffmannH10a,HAH12b}).

Let $\sigma$ be a (normalised) substitution and let $f(x_1,\dots,x_n)$ be a term. 
It follows from the definitions that $f(x_1\sigma,\dots,x_n\sigma) \rssrew v$
iff $\eval{\sigma}{}{f(x_1,\dots,x_n)}{v}$. More, precisely we have the
following proposition. 

\begin{proposition}
\label{p:1}
Let $f$ be a defined function symbol of arity $n$ and $\sigma$
a substitution. Then $\eval{\sigma}{m}{f(x_1,\dots,x_n)}{v}$  
holds iff $\dheight(f(x_1\sigma,\dots,x_n\sigma),\rsrew) = m$ holds.
\end{proposition}
\begin{proof}
In proof of the direction from left to right, we show the stronger statement
that $\eval{\sigma}{m}{t}{v}$ implies $\dheight(t\sigma,\rsrew) = m$ by induction
on the size of the proof of the judgement $\eval{\sigma}{m}{f(x_1,\dots,x_n)}{v}$. 
For the opposite direction, we show that if $\dheight(t\sigma,\rsrew) = m$, then 
$\eval{\sigma}{m}{t}{v}$ by induction on the length of the derivation 
$D\colon t\sigma \rssrew v$.
\end{proof}

The next (technical) lemma follows by a straightforward inductive
argument.
\begin{lemma}
\label{l:4}
Let $t$ be a term, let $v$ be a value and
let $\sigma$ be a substitution. If $\eval{\sigma}{m}{t}{v}$
and if $\sigma'$ is an extension of $\sigma$, then $\eval{\sigma'}{m}{t}{v}$. 
Furthermore the sizes of the derivations of the corresponding judgements are the same.
\end{lemma}

\section{Annotated Types}
\label{AnnotatedTypes}

Let $S$ be a set of types. We call a type $A \in S$ \emph{annotated}, if $A$
is decorated with resource annotation. These annotations will allow us to read off 
the potential of a well-typed term $t$ from the annotations.

\begin{definition}
\label{d:simpletypes}
Let $S$ be a set of types.
An \emph{annotated type} $\atyp{A}{\vec{p}}$, is a pair
consisting of a type $A \in S$ and a vector $\vec{p}=(p_1,\dots,p_k)$
over non-negative rational numbers, typically natural numbers. The vector $\vec{p}$ is called \emph{resource annotation}.
\end{definition}

Resource annotations are denoted by $\vec{p}$, $\vec{q}$, $\vec{u}$, $\vec{v}$, \dots, 
possibly extended by subscripts and we write $\Vecs$ for the set of such annotations. 
For resource annotations $(p)$ of length $1$ we write $p$. The empty annotation $()$ is written $0$. 
We will see that a resource annotation does not change its meaning if zeroes are appended at the end, 
so, conceptually, we can identify $()$ with $(0)$. If $\vec p=(p_1,\dots,p_k)$ 
we write $k=\len{\vec{p}}$ and $\max\vec{p}=\max_i p_i$. 
We define the notations  $\vec p\leqslant \vec q$ and $\vec p+\vec q$ and $\lambda \vec p$ for 
$\lambda \geqslant 0$ component-wise, filling up with $0$s if needed. 
So, for example $(1,2)\leqslant (3,4,5)$ and $(1,2)+(3,4,5)=(4,6,5)$. 
Furthermore, we recall the additive shift \cite{HoffmannH10a} given by 
\begin{equation*}
 \shift(\vec{p}) \defsym (p_1 + p_2, p_2 + p_3, \dots, p_{k-1} + p_k, p_k)
 \tpkt
\end{equation*}
We also define the interleaving $\vec p\interleave \vec q$ by $(p_1,q_1,p_2,q_2,$ $\dots,p_k,q_k)$ 
where, as before the shorter of the two vectors is padded with $0$s. 
Finally, we use the notation $\Diamond\vec{p} = p_1$ for the first entry of an annotation vector. 

If no confusion can arise, we refer to annotated types simply as types. 
In contrast to Hoffmann et al.~\cite{HoffmannH10a,Hoffmann:2011}, we generalise the
concept of annotated types to arbitrary (data) types. In~\cite{HoffmannH10a} only
list types, in~\cite{Hoffmann:2011} list and tree types have been annotated. 

\begin{definition}
Let $\FS$ be a signature.
Suppose $\FS(f) =\typdcl{A_1 \times \cdots \times A_n}{C}$, such
that the $A_i$ ($i=1,\dots,n$) and $C$ are types. Consider the annotated
types $A_i^{\vec{u_i}}$ and $\atyp{A}{\vec{v}}$. Then 
an \emph{annotated type declaration} for $f$ is a type declaration over
annotated types, decorated with a number $p$:
\begin{equation*}
  \atypdcl{A_1^{\vec{u_1}} \times \cdots \times A_n^{\vec{u_n}}}{\atyp{C}{\vec{v}}}{p}
    \tpkt
\end{equation*}
The set of annotated type declarations is denoted as $\TDannot$.
\end{definition}

% notation
We write $A^0$ instead of $\atypdcl{}{\atyp{A}{0}}{0}$.
We lift signatures to \emph{annotated signatures} 
$\FS \colon \CS \cup \DS \to (\pow(\TDannot)\setminus \varnothing)$ 
by mapping a function symbol to a non-empty set of annotated type 
declarations. Hence for any $f \in \CS \cup \DS$ we allow multiple types. 
If $f$ has result type $C$, then for each
annotation $C^{\vec q}$ there should exist exactly one declaration of
the form 
$\atypdcl{A_1^{\vec{p_1}} \times \cdots \times A_n^{\vec{p_n}}}{C^{\vec q}}{p}$ in $\FS(f)$.
Moreover, constructor annotations are to satisfy the
\emph{superposition principle}: If a constructor $c$ admits the
annotations 
$\atypdcl{A_1^{\vec{p_1}} \times \cdots \times A_n^{\vec{p_n}}}{C^{\vec{q}}}{p}$ and 
$\atypdcl{A_1^{\vec{p'_1}} \times \cdots \times A_n^{\vec{p'_n}}}{C^{\vec{q'}}}{p'}$ 
then it also has the annotations 
$\atypdcl{A_1^{\lambda \vec{p_1}} \times \cdots \times A_n^{\lambda \vec{p_n}}}{C^{\lambda \vec{q}}}{\lambda p}$ ($\lambda\geqslant 0$) and
$\atypdcl{A_1^{\vec{p_1}+\vec{p'_1}} \times \cdots \times A_n^{\vec{p_n}+\vec{p'_n}}}{C^{\vec{q}+\vec{q'}}}{p+p'}$.

Note that, in view of superposition and uniqueness, the annotations of a given constructor are 
uniquely determined once we fix the annotated types 
for result annotations of the form $(0,\dots,0,1)$ 
(remember the implicit filling up with $0$s). 
An annotated signature $\FS$ is simply called signature, where we sometimes
write $\typed{f}{\atypdcl{A_1 \times \cdots \times A_n}{C}{p}}$ instead of $\atypdcl{A_1 \times \cdots \times A_n}{C}{p} \in \FS(f)$. 

\begin{example}[continued from Example~\ref{ex:2}]
\label{ex:3}
In order to extend  $\FS$ to an annotated signature we can set 
\begin{alignat*}{3}
  \FS(\zero) &\defsym \{\aNat{\vec p}\mid \vec p\in \Vecs\} & \hspace{.5ex}
  \FS(\mS) &\defsym \{\atypdcl{\aNat{\shift(\vec p)}}{\aNat{\vec p}}{\Diamond\vec p}\mid \vec p\in \Vecs\}
  \\
  \FS(\nil) &\defsym \{ \aList{\vec p} \mid \vec p\in \Vecs\} & \hspace{.5ex}
  \FS(\cons) &\defsym \{\atypdcl{\aNat{0} \times \aList{\shift(\vec p)}}{\aList{\vec p}}{\Diamond\vec p} \mid \vec p\in \Vecs\}
  \\
  \FS(\queue) &\defsym \{\atypdcl{\aList{\vec p} \times \aList{\vec q}}{\aQueue{\vec p\interleave\vec q}}{0} \mid \vec p,\vec q\in \Vecs\}
\end{alignat*}
In particular, we have the typings 
$\cons : \atypdcl{\aNat{0} \times \aList{7}}{\aList{7}}{7}$ and 
$\cons : \atypdcl{\aNat{0} \times \aList{(10,7)}}{\aList{(3,7)}}{3}$ and 
$\queue : \atypdcl{\aList{1} \times \aList{3}}{\aQueue{(1,3)}}{0}$. 

We omit annotations for the defined symbols and refer to Example~\ref{ex:5} for a complete signature with different constructor annotations. 
% Finally, we define the (annotated) type declaration of the defined function
% $\snoc$:
% %
% $\FS(\snoc) \defsym \{\atypdcl{\aQueue{(p,k,q,l)} \times \aNat{0}}{\aQueue{(p,k,q,l)}}{4} \}$
% %
% or $\typed{\snoc}{\atypdcl{\aQueue{(p,k,q,l)} \times \aNat{0}}{\aQueue{(p,k,q,l)}}{4}}$ for short.
\end{example}

The next definition introduces the notion of the potential
of a value. 

\begin{definition}
\label{d:potential}
Let $v = c(v_1,\dots,v_n) \in \TA(\CS)$ and let 
$\atypdcl{A_1 \times \cdots \times A_n}{C}{p} \in \FS(c)$.
Then the \emph{potential} of $v$ is defined inductively as
\begin{equation*}
  \Phi(\typed{v}{C}) \defsym p + \Phi(\typed{v_1}{A_1}) + \cdots + \Phi(\typed{v_n}{A_n})
  \tpkt
\end{equation*}
\end{definition}

Note that by assumption the declaration in $\FS(c)$ is unique. 

\begin{example}[continued from Example~\ref{ex:3}]
\label{ex:4}
It is easy to see that for any term $t$ of type $\aNat{0}$, 
we have $\Phi(\typed{t}{\aNat{0}}) = 0$ and $\Phi(\typed{t}{\aNat{\lambda}})=\lambda t$. 

If $l$ is a list then 
$\Phi(\typed{l}{\aList{(p,q)}}) = p \cdot \len{l} + q \cdot \binom{\len{l}}{2}$.
where $\len{l}$ denotes the length of $l$, that is the number of $\cons$ in $l$.
Let $\len{l} = \ell$. We proceed by induction on $\ell$. Let $\ell = 0$. Then
$\Phi(\typed{\nil}{\aList{(p,q)}}) = 0$ as required. Suppose $\ell = \ell' + 1$:
\begin{align*}
  \Phi(\typed{n \cons l'}{\aList{(p,q)}}) &= p + \Phi(\typed{n}{\aNat{0}}) + 
  \Phi(\typed{l'}{\aList{(p+q,q)}}) 
  \\
  &= p + (p+q) \cdot \ell' + q \cdot \binom{\ell'}{2} 
  \\
  &= p \cdot \ell + q \cdot \left[ \binom{\ell'}{1} + \binom{\ell'}{2} \right]
  = p \cdot \ell + q \cdot \binom{\ell}{2} \tpkt
\end{align*} 
More generally, we have $\Phi(\typed{l}{\aList{\vec p}})=\sum_i p_i\binom{\len{l}}{i}$. 
Finally, if $\queue(l,k)$ has type $\Queue$ then 
$\Phi(\typed{\queue(l,k)}{\aQueue{\vec{p} \interleave \vec{q}}})=\Phi(\typed{l}{\aList{\vec p}})+ \Phi(\typed{k}{\aList{\vec q}})$. 
\end{example}

The \emph{sharing relation} 
$\share{A^{\vec{p}}}{A_1^{\vec{p_1}},A_2^{\vec{p_2}}}$ holds if $A=A_1=A_2$ and 
$\vec{p_1} + \vec{p_2} = \vec p$.
The subtype relation is defined  as follows: 
$\atyp{A}{\vec{p}} \subtype \atyp{B}{\vec{q}}$, if $A = B$ and
$\vec{p} \geqslant \vec{q}$.

\begin{lemma}
\label{l:10}
If $\share{A^{\vec{p}}}{A_1^{\vec{p_1}},A_2^{\vec{p_2}}}$ then 
$\Phi(\typed{v}{A^{\vec{p}}}) = \Phi(\typed{v}{A_1^{\vec{p_1}}}) + \Phi(\typed{v}{A_2^{\vec{p_2}}})$ 
holds for any value of type $A$. If $\atyp{A}{\vec{p}} \subtype \atyp{B}{\vec{q}}$ then 
$\Phi(\typed{v}{\atyp{A}{\vec p}})\geqslant \Phi(\typed{v}{\atyp{B}{\vec q}})$ again for any $v:A$. 
\end{lemma}
\begin{proof}
The proof of the first claim is by induction on the structure of $v$. 
We note that by superposition together with uniqueness the additivity property propagates 
to the argument types. For example, if we have the annotations 
$\mS:\atypdcl{\aNat{2}}{\aNat{3}}{4}$ and 
$\mS:\atypdcl{\aNat{4}}{\aNat{5}}{6}$ and $\mS:\atypdcl{\aNat{x}}{\aNat{y}}{10}$ 
then we can conclude $x=6$, $y=8$, for this annotation must be present by superposition 
and there can only be one by uniqueness. 

The second claim follows from the first one and nonnegativity of potentials. 
\end{proof}

\begin{figure}[t]
  \centering
  \begin{tabular}{c}
    $\infer{\tjudge{\typed{x_1}{A_1^{\vec{u_1}}},\dots,\typed{x_n}{A_n^{\vec{u_n}}}}{p}{\typed{f(x_1,\dots,x_n)}{\atyp{C}{\vec{v}}}}}{%
      f \in \CS \cup \DS
      &
      \atypdcl{A_1^{\vec{u_1}} \times \cdots \times A_n^{\vec{u_n}}}{\atyp{C}{\vec{v}}}{p} \in \FS(f)
    }$
    \qquad \hfill \qquad 
    $\infer{\tjudge{\Gamma}{p'}{\typed{t}{C}}}{%
        \tjudge{\Gamma}{p}{\typed{t}{C}}
        &
        p' \geqslant p
        }$
    \\[2.5ex]
    $\infer{\tjudge{\Gamma_1,\dots,\Gamma_n}{p}{\typed{f(t_1,\dots,t_n)}{C}}}{%
      \begin{minipage}[b]{40ex}
        all $x_i$ are fresh\\[.5ex]
        $\tjudge{\typed{x_1}{A_1},\dots,\typed{x_n}{A_n}}{p_0}{\typed{f(x_1,\dots,x_n)}{C}}$
      \end{minipage}
      &
      \begin{minipage}[b]{32ex}
        $p = \sum_{i=0}^n p_i$\\[.5ex]
        $\tjudge{\Gamma_1}{p_1}{\typed{t_1}{A_1}} \ \cdots \
        \tjudge{\Gamma_n}{p_n}{\typed{t_n}{A_n}}$
      \end{minipage}
      }$%
    \\[2.5ex] 
    $\infer{\tjudge{\Gamma, \typed{x}{A}}{p}{\typed{t}{C}}}{%
      \tjudge{\Gamma}{p}{\typed{t}{C}}
      }$
    \qquad \hfill \qquad 
    $\infer{\tjudge{\Gamma, \typed{z}{A}}{p}{\typed{t[z,z]}{C}}}{%
      \tjudge{\Gamma, \typed{x}{A_1}, \typed{y}{A_2}}{p}{\typed{t[x,y]}{C}}%
      & 
      \share{A}{A_1,A_2}
      &
      \text{$x$, $y$ are fresh}
      }$
    \\[2.5ex]
    $\infer{\tjudge{\Gamma, \typed{x}{A}}{p}{\typed{t}{C}}}{%
      \tjudge{\Gamma, \typed{x}{B}}{p}{\typed{t}{C}}
      &
      A \subtype B
      }$
    \hfill 
    $\infer{\tjudge{\typed{x}{A}}{0}{\typed{x}{A}}}{}$
    \hfill
    $\infer{\tjudge{\Gamma}{p}{\typed{t}{C}}}{%
      \tjudge{\Gamma}{p}{\typed{t}{D}}
      &
      D \subtype C
      }$      
  \end{tabular}
  \caption{Type System for Rewrite Systems}
  \label{fig:2}
\end{figure}

The set of typing rules for TRSs are given in Figure~\ref{fig:2}. 
Observe that the type system employs the assumption that $\RS$ is
left-linear. In a nutshell, the method works as follows:
Let $\Gamma$ be a typing context and let us consider the typing judgement 
$\tjudge{\Gamma}{p}{\typed{t}{A}}$ derivable from the type rules. 
Then $p$ is an upper-bound to the amortised cost 
required for reducing $t$ to a value. 
The derivation height of $t\sigma$ (with respect to innermost rewriting) is
bound by the difference in the potential before and after the
evaluation plus $p$. Thus if the sum of the potential of the arguments of
$t\sigma$ is in $\bigO(n^k)$, where $n$ is the size of the arguments, then
the runtime complexity of $\RS$ lies in $\bigO(n^k)$.

Recall that any rewrite rule $l \to r \in \RS$ 
can be written as $f(l_1,\dots,l_n) \to r$ with $l_i \in \TA(\CS,\VS)$. 
We introduce \emph{well-typed} TRSs.

\begin{definition}
\label{d:welltyped}
Let $f(l_1,\dots,l_n) \to r$ be a rewrite rule in $\RS$
and let $\Var(f(\vec{l})) = \{y_1,\dots,y_\ell\}$.
Then $f \in \DS$ is \emph{well-typed} wrt.~$\FS$, if we obtain
\begin{equation}
\label{eq:welltyped}
\tjudge{\typed{y_{1}}{B_{1}},\dots,\typed{y_{\ell}}{B_{\ell}}}{p - 1 + \sum_{i=1}^n k_i}{\typed{r}{C}}
  \tkom
\end{equation}
for all $\atypdcl{A_1 \times \cdots \times A_n}{C}{p} \in \FS(f)$,
for all types $B_{j}$ ($j \in \{1,\dots,\ell\}$), 
% GM added 
and all costs $k_i$, such that
$\tjudge{\typed{y_{1}}{B_{1}},\dots,\typed{y_{\ell}}{B_{\ell}}}{k_i}{\typed{l_i}{A_i}}$
is derivable. A TRS $\RS$ over $\FS$ is \emph{well-typed} if any defined $f$ is well-typed.
\end{definition}

Contrary to analogous definitions in the literature on amortised resource analysis 
the definition recurs to the type system in order to specify the available resources 
in the type judgement~\eqref{eq:welltyped}. This is necessary to adapt amortised
analysis to rewrite systems.

Let $\Gamma$ be a typing context and let $\sigma$ be a substitution. We call
$\sigma$ \emph{well-typed (with respect to $\Gamma$)} if for all $x \in \dom(\Gamma)$
$x\sigma$ is of type $\Gamma(x)$. 
We extend Definition~\ref{d:potential} to substitutions $\sigma$ 
and typing contexts $\Gamma$. Suppose $\sigma$ is well-typed 
with respect to $\Gamma$. Then $\Phi(\typed{\sigma}{\Gamma}) \defsym \sum_{x \in \dom(\Gamma)} 
\Phi(\typed{x\sigma}{\Gamma(x)})$.
We state and prove our first soundness result.

\begin{theorem}
\label{l:2}
Let $\RS$ and $\sigma$ be well-typed. Suppose $\tjudge{\Gamma}{p}{\typed{t}{A}}$
and $\eval{\sigma}{m}{t}{v}$. Then 
$\Phi(\typed{\sigma}{\Gamma}) - \Phi(\typed{v}{A}) + p  \geqslant m$.
% %
% Further, if for all values $v$ and types $A$: $\Phi(\typed{v}{A}) \in \bigO(n^k)$, 
% where $n = \size{v}$, then $\rc_{\RS}(n) \in \bigO(n^k)$.
\end{theorem}
\begin{proof}
Let $\Pi$ be the proof deriving $\eval{\sigma}{m}{t}{v}$
and let $\Xi$ be the proof of $\tjudge{\Gamma}{p}{\typed{t}{A}}$. 
The proof of the theorem proceeds by main-induction on the length of $\Pi$
and by side-induction on the length of $\Xi$.
\begin{enumerate}
\item Suppose $\Pi$ has the form 
\begin{equation*}
    \infer{\eval{\sigma}{m}{x}{v}}{x\sigma = v}
    \tkom
\end{equation*}
such that $t = x$ and $v = x\sigma$. Wlog.\ $\Xi$ is of form
$\tjudge{\typed{x}{A}}{0}{\typed{x}{A}}$. Then 
$\Phi(\typed{\sigma}{\Gamma}) = \Phi(\typed{x\sigma}{A})$ and the
theorem follows.

\item Suppose $\Pi$ has the form 
\begin{equation*}
    \infer{\eval{\sigma}{m}{c(x_1,\dots,x_n)}{c(v_1,\dots,v_n)}}{%
       c \in \CS
       &
       x_1\sigma = v_1 
       & \cdots
       & x_n\sigma = v_n%
     }
    \tkom
\end{equation*}
such that $t = c(x_1,\dots,x_n)$ and $v = c(v_1,\dots,v_n)$. Further wlog.\ we suppose 
that $\Xi$ ends in the following judgement:
\begin{equation*}
    \tjudge{\typed{x_1}{A_1^{\vec{u_1}}},\dots,\typed{x_n}{A_n^{\vec{u_n}}}}{p}{\typed{c(x_1,\dots,x_n)}{\atyp{C}{\vec{w}}}}
    \tpkt
\end{equation*}
Then we have 
$\atypdcl{A_1^{\vec{u_1}} \times \cdots \times A_n^{\vec{u_n}}}{\atyp{C}{\vec{w}}}{p} \in \FS(c)$
and thus:
\begin{equation*}
  \Phi(\typed{\sigma}{\Gamma}) + p = 
  p  + \sum_{i=1}^n \Phi(\typed{x_i\sigma}{A_i^{\vec{u_i}}}) 
  = p + \sum_{i=1}^n \Phi(\typed{v_i}{A_i^{\vec{u_i}}}) 
  = \Phi(\typed{c(v_1,\dots,v_n)}{C^{\vec{w}}}) 
  \tkom
\end{equation*}
from which the theorem follows.

\item Suppose $\Pi$ ends in the following rule:
  \begin{equation*}
    \infer{\eval{\sigma}{m+1}{f(x_1,\dots,x_n)}{v}}{%
      \exists \ f(l_1,\dots,l_n) \to r \in \RS
       & \exists \tau \ \forall i\colon x_i\sigma = l_i\tau
       & \eval{\sigma \dunion \tau}{m}{r}{v}
     }
     \tpkt
  \end{equation*}
Then $t=f(x_1,\dots,x_n)$ and $f(x_1,\dots,x_n)\sigma = f(l_1,\dots,l_n)\tau$.
Suppose $\Var(f(\vec{l})) = \{y_{1},\dots,y_{\ell}\}$ and let
$\Var(l_i) = \{y_{i1},\dots,y_{il_i}\}$ for $i \in \{1,\dots,n\}$. 
As $\RS$ is left-linear we have $\Var(f(l_1,\dots,l_n)) = \biguplus_{i=1}^n \Var(l_i)$.
We set  $\Gamma = \typed{x_1}{A_1},\dots,\typed{x_n}{A_n}$.
By the assumption $\tjudge{\Gamma}{p}{\typed{t}{A}}$ and 
well-typedness of $\RS$ we obtain 
\begin{equation*}
  \tjudge{\overbrace{\typed{y_1}{B_1}, \dots, \typed{y_{\ell}}{B_\ell}}^{{} =: \Delta}}{p - 1 + \sum_{i=1}^n k_i}{\typed{r}{C}}
    \tkom
\end{equation*}
as in~\eqref{eq:welltyped}. By main induction hypothesis together with the
above equation, we have 
$\Phi(\typed{\sigma \dunion \tau}{\Delta}) - \Phi(\typed{v}{C}) + 
  p - 1 + \sum_{i=1}^n k_i  \geqslant m$.
Furthermore, we have 
\begin{align*}
  \Phi(\typed{\sigma}{\Gamma}) 
  & = \sum_{i=1}^n \Phi(\typed{x_i\sigma}{A_i}) 
  = \sum_{i=1}^n \left( k_i + \Phi(\typed{y_{i1}\tau}{B_{i1}}) + \cdots + 
                        \Phi(\typed{y_{il_i}\tau}{B_{il_i}}) \right) \\
  & = \Phi(\typed{\sigma \dunion \tau}{\Delta}) + \sum_{i=1}^n k_i  \tpkt
\end{align*}
Here the first equality follows by an inspection on the case for the constructors. 
In sum, we obtain 
\begin{equation*}
  \Phi(\typed{\sigma}{\Gamma}) - \Phi(\typed{v}{C}) + p =
  \Phi(\typed{\sigma \dunion \tau}{\Delta}) + \sum_{i=1}^n k_i - \Phi(\typed{v}{C}) + p 
  \geqslant m+1
  \tkom
\end{equation*}
from which the theorem follows.

\item Suppose the last rule in $\Pi$ has the form
\begin{equation*}
    \infer{\eval{\sigma}{m}{f(t_1,\dots,t_n)}{v}}{%
      \eval{\sigma \dunion \rho}{m_0}{f(x_1,\ldots,x_n)}{v}
      & \eval{\sigma}{m_1}{t_1}{v_1}
      & \cdots
      & \eval{\sigma}{m_n}{t_n}{v_n}
      & m = \sum_{i=0}^n m_i% 
    }
    \tpkt
\end{equation*}
We can assume that $t$ is linear, compare
the case employing the share operator. Hence the last rule in the type inference $\Xi$ is of the following
form. 
\begin{equation*}
  \infer{\tjudge{\Gamma_1,\dots,\Gamma_n}{p}{\typed{f(t_1,\dots,t_n)}{C}}}{%
      \tjudge{\overbrace{\typed{y_1}{A_1},\dots,\typed{y_n}{A_n}}^{{} =: \Delta}}{p_0}{%
             \typed{f(\vec{y})}{C}}
      & \tjudge{\Gamma_1}{p_1}{\typed{t_1}{A_1}}
      & \cdots
      & \tjudge{\Gamma_n}{p_n}{\typed{t_n}{A_n}}
      & p = \sum_{i=0}^n p_i%
    }
    \tpkt
\end{equation*}
By induction hypothesis:
$\Phi(\typed{\sigma}{\Gamma_i}) - \Phi(\typed{v_i}{A_i}) + p_i \geqslant m_i$
for all $i=1,\dots,n$.
Hence 
\begin{equation}
\label{eq:6}
  \sum_{i=1}^n \Phi(\typed{\sigma}{\Gamma_i}) - \sum_{i=1}^n \Phi(\typed{v_i}{A_i}) + \sum_{i=1}^n p_i
  \geqslant \sum_{i=1}^n m_i \tpkt
\end{equation}
Again by induction hypothesis we obtain:
\begin{equation}
\label{eq:7}
  \Phi(\typed{\sigma \dunion \rho}{\Delta}) - \Phi(\typed{v}{C}) + p_0 \geqslant
  m_0 \tpkt
\end{equation}
Now $\Phi(\typed{\sigma}{\Gamma}) = \sum_{i=1}^n \Phi(\typed{\sigma}{\Gamma_i})$
and $\Phi(\typed{\sigma \dunion \rho}{\Delta}) = \Phi(\typed{\rho}{\Delta}) =
\sum_{i=1}^n \Phi(\typed{v_i}{A_i})$. 
Due to~\eqref{eq:6} and~\eqref{eq:7}, we obtain 
\begin{align*}
  \Phi(\typed{\sigma}{\Gamma}) + \sum_{i=0}^n p_i & = 
  \sum_{i=1}^n \Phi(\typed{\sigma}{\Gamma_i}) + \sum_{i=1}^n p_i + p_0\\
  & \geqslant \sum_{i=1}^n \Phi(\typed{v_i}{A_i}) + \sum_{i=1}^n m_i + p_0 
  \geqslant \Phi(\typed{v}{C}) + \sum_{i=0}^n m_i 
  \tkom
\end{align*}
and thus $\Phi(\typed{\sigma}{\Gamma} - \Phi(\typed{v}{C}) + p \geqslant
m$.

\item Suppose $\Xi$ is of form
  \begin{equation*}
    \infer{\tjudge{\Gamma}{p'}{\typed{t}{C}}}{%
        \tjudge{\Gamma}{p}{\typed{t}{C}}
        &
        p' \geqslant p
        }
        \tpkt
  \end{equation*}
By side-induction on $\tjudge{\Gamma}{p}{\typed{t}{C}}$ together
with $\eval{\sigma}{m}{t}{v}$ we conclude
$\Phi(\typed{\sigma}{\Gamma}) - \Phi(\typed{v}{A}) + p \geqslant m$.
Then the theorem follows from the assumption $p' \geqslant p$. 

\item Suppose $\Xi$ is of form
  \begin{equation*}
    \infer{\tjudge{\Gamma, \typed{x}{A}}{p}{\typed{t}{C}}}{%
      \tjudge{\Gamma}{p}{\typed{t}{C}}
      }
      \tpkt
  \end{equation*}
We conclude by side-induction together
with $\eval{\sigma}{m}{t}{v}$ we conclude
$\Phi(\typed{\sigma}{\Gamma}) - \Phi(\typed{v}{A}) + p \geqslant m$. Clearly
$\Phi(\typed{\sigma}{\Gamma, \typed{x}{A}}) \geqslant \Phi(\typed{\sigma}{\Gamma})$
and the theorem follows.

\item Suppose $\Xi$ is of form
  \begin{equation*}
    \infer{\tjudge{\Gamma, \typed{z}{A}}{p}{\typed{t[z,z]}{C}}}{%
      \tjudge{\Gamma, \typed{x}{A_1}, \typed{y}{A_2}}{p}{\typed{t[x,y]}{C}}%
      & 
      \share{A}{A_1,A_2}
      }
  \end{equation*}
By assumption $\eval{\sigma}{m}{t[z,z]}{v}$; let
$\rho \defsym \sigma \dunion \{x \mapsto z\sigma, y \mapsto z\sigma\}$.
As $\eval{\sigma}{m}{t[z,z]}{v}$, we obtain 
$\eval{\rho}{m}{t[x,y]}{v}$ by definition. From the side-induction on 
$\tjudge{\Gamma, \typed{x}{A_1}, \typed{y}{A_2}}{p}{\typed{t[x,y]}{C}}$
and $\eval{\rho}{m}{t[x,y]}{v}$ we conclude that 
\begin{equation*}
 \Phi(\typed{\rho}{\Gamma,\typed{x}{A_1}, \typed{y}{A_2}}) - \Phi(\typed{v}{C} + p \geqslant m 
 \tpkt
\end{equation*}
The theorem follows as by definition of $\rho$ and Lemma~\ref{l:10}, we obtain
\begin{equation*}
  \Phi(\typed{\sigma}{\Gamma,\typed{z}{A}}) = 
  \Phi(\typed{\rho}{\Gamma, \typed{x}{A_1}, \typed{y}{A_2}})
  \tpkt
\end{equation*}

\item Suppose $\Xi$ is of form
\begin{equation*}
    \infer{\tjudge{\Gamma, \typed{x}{A}}{p}{\typed{t}{C}}}{%
      \tjudge{\Gamma, \typed{x}{B}}{p}{\typed{t}{C}}
      &
      A \subtype B
      }
\end{equation*}
By assumption $\eval{\sigma}{m}{t}{v}$ and by induction hypothesis
$\Phi(\typed{\sigma}{\Gamma, \typed{x}{B}}) - \Phi(\typed{v}{A}) + p \geqslant
m$. 
By definition of the subtype relation 
$\Phi(\typed{x\sigma}{A}) \geqslant \Phi(\typed{x\sigma}{B})$. Hence
the theorem follows. 

\item Suppose $\Xi$ is of form
\begin{equation*}
    \infer{\tjudge{\Gamma}{p}{\typed{t}{C}}}{%
      \tjudge{\Gamma}{p}{\typed{t}{D}}
      &
      D \subtype C
      }
\end{equation*}
The case follows similarly to the sub-case before by induction
hypothesis.
From this the theorem follows.
\end{enumerate}

The second assertion of the theorem follows from the first together with 
the assumption that every defined symbol in $\FS$ is well-typed and 
Proposition~\ref{p:1}.
\end{proof}

\begin{example}[continued from Example~\ref{ex:1}]
\label{ex:5}  
Consider the TRS $\RSa$ from Example~\ref{ex:1}. We detail the signature $\FS$,
starting with the constructor symbols. 
\begin{alignat*}{4}
  \zero\colon\!&\aNat{p} & \hspace{2ex} 
  \mS\colon\!&\atypdcl{\aNat{p}}{\aNat{p}}{p} & \hspace{2ex}
  \errorHead\colon\!& \aNat{p} & \hspace{2ex}
  \queue\colon\!&\atypdcl{\aList{p} \times \aList{q}}{\aQueue{(p,q)}}{0}
  \\
  \nil\colon\!&\aList{q} & \hspace{2ex} 
  \cons\colon\!& \atypdcl{\aNat{0} \times \aList{q}}{\aList{q}}{q} & \hspace{2ex} 
  \errorTail\colon\!&\aQueue{(0,1)} \tkom
\end{alignat*}
where $p,q \in \N$. Furthermore we make use of the following types for defined
symbols.
\begin{alignat*}{4}
  \checkF\colon &\atypdcl{\aQueue{(0,1)}}{\aQueue{(0,1)}}{3} & \hspace{2.5ex}
  \tail\colon &\atypdcl{\aQueue{(0,1)}}{\aQueue{(0,1)}}{4} & \hspace{2.5ex}
  \head\colon &\atypdcl{\aQueue{(0,1)}}{\aNat{0}}{1}
  \\
  \revp\colon &\atypdcl{\aList{1} \times \aList{0}}{\aList{0}}{1} & \hspace{2.5ex}
  \rev\colon &\atypdcl{\aList{1} \times \aList{0}}{\aList{0}}{2} & \hspace{2.5ex}
  \\
  \snoc\colon &\atypdcl{\aQueue{(0,1)} \times \aNat{0}}{\aQueue{(0,1)}}{5} & \hspace{2.5ex}
  \enq\colon &\atypdcl{\aNat{6}}{\aQueue{(0,1)}}{1} \tkom \hspace{2.5ex}
\end{alignat*}
Let $\FS$ denote the induced signature. Based on the above definitions
it is not difficult to verify that $\RSa$ is well-typed wrt.~$\FS$. 
We show that $\enq$ is well-typed. 
Consider rule~6. First, we observe that $6$ resource units become available
for the recursive call, as 
$\tjudge{\typed{n}{\aNat{6}}}{6}{\typed{\mS(n)}{\aNat{6}}}$ is derivable.
Second, we have the following partial type derivation; missing
parts are easy to fill in.
\begin{equation*}
  \infer{\tjudge{\typed{n}{\aNat{6}}}{6}{\typed{\snoc(\enq(n),n)}{\aQueue{(0,1)}}}}{%
    \infer{\tjudge{\typed{n_1}{\aNat{6}},\typed{n_2}{\aNat{0}}}{6}{\typed{\snoc(\enq(n_1),n_2)}{\aQueue{(0,1)}}}}{%
      \tjudge{\typed{q}{\aQueue{(0,1)}},\typed{m}{\aNat{0}}}{5}{\typed{\snoc(q,m)}{\aQueue{(0,1)}}}
      &
      \begin{minipage}[b]{35ex}
        \mbox{} \hfill $\tjudge{\typed{n_2}{\aNat{0}}}{0}{\typed{n_2}{\aNat{0}}}$\\[1ex]
        $\tjudge{\typed{n_1}{\aNat{6}}}{1}{\typed{\enq(n_1)}{\aQueue{(0,1)}}}$
      \end{minipage}
    }
  }
\end{equation*}
Considering rule~7, it is easy to see that 
$\tjudge{\typed{n}{\aNat{6}}}{0}{\typed{\queue(\nil,\nil)}{\aQueue{(0,1)}}}$
is derivable. Thus $\enq$ is well-typed and we conclude optimal linear runtime complexity of $\RSa$.
\end{example}

\paragraph{Polynomial bounds} 

Note that if the type annotations are chosen such that for each type
$A$ we have $\Phi(\typed{v}{A}) \in \bigO(n^k)$ for $n=\len{v}$ then
$\rc_{\RS}(n) \in \bigO(n^k)$ as well. The following proposition gives
a sufficient condition as to when this is the case and in particular
subsumes the type system in \cite{HoffmannH10a}.

\begin{theorem}\label{polyt}
Suppose that for each constructor $c$ with $\atypdcl{A_1^{\vec{u_1}} \times \cdots \times A_n^{\vec{u_n}}}{C^{\vec{w}}}{p} \in \FS(c)$, there exists $\vec{r}_i \in \Vecs$ such that 
$\vec{u_i} \leqslant \vec{w} + \vec{r}_i$ where $\max{\vec r}_i \leqslant \max \vec{w}=:r$ and 
$p \leqslant r$ with $\len{\vec{r}_i} < \len{\vec w}=:k$. 
Then $\Phi(\typed{v}{C^{\vec w}}) \leqslant r\len{v}^k$. 
\end{theorem}
\begin{proof}
The proof is by induction on the size of $v$. Note that, if $k=0$ then 
$\Phi(\typed{v}{C^{\vec{w}}})=0$. This follows by 
superposition and uniqueness. Otherwise,  we have
\begin{align*}
  \Phi(\typed{c(v_1,\dots,v_n)}{C^{\vec w}}) & \leqslant 
  r+\Phi(\typed{v_1}{A_1^{\vec{w} +\vec{r}_1}})+\dots+\Phi(\typed{v_n}{A_n^{\vec{w}+\vec{r}_n}})
  \\
  & \leqslant r(1+\len{v_1}^k+\len{v_1}^{k-1}+\dots+\len{v_n}^k+\len{v_n}^{k-1})
  \\
  & \leqslant r(1+\len{v_1}+\dots+\len{v_n})^k=r\len{v}^k
  \tpkt
\end{align*}
Here we employ Lemma~\ref{l:10} to conclude for all $i=1,\dots,n$: 
\begin{equation*}
  \Phi(\typed{v_i}{A_i^{\vec{w} +\vec{r}_i}}) = \Phi(\typed{v_i}{A_i^{\vec{w}}}) + \Phi(\typed{v_i}{A_i^{\vec{r}_i}})
  \tpkt
\end{equation*}
Based on this observation we apply induction hypothesis to obtain
the second line. Furthermore in the last line we employ the
multinomial theorem. 
\end{proof}

We note that our running example satisfies the premise to the
proposition.  In concrete cases more precise bounds than those given
by Theorem~\ref{polyt} can be computed as has been done in Example~\ref{ex:4}.
The next example clarifies that potentials are not restricted to polynomials.

\begin{example}
Consider that we annotate the constructors for natural numbers as 
$\typed{\zero}{\aNat{\vec p}}$ and 
$\typed{\mS}{\atypdcl{\aNat{2\vec{p}}}{\aNat{\vec{p}}}{\Diamond\vec{p}}}$. 
We then have, for example, $\Phi(\typed{t}{\aNat{1}})=2^{t+1}-1$.
\end{example}

As mentioned in the introduction, foundational issues are our main
concern.  However, the potential-based method detailed above seems
susceptible to automation.  One conceives the resource annotations as
variables and encodes the constraints of the typing rules in
Figure~\ref{fig:2} over these resource variables.

\section{Small-Step Semantics}
\label{SmallStep}

\begin{figure}[t]
  \centering
  \begin{tabular}{c}
   $\infer{\smallstep{0}{x}{\sigma}{v}{\sigma}}{%
     x\sigma = v%
   }$
   \hfill
   $\infer{\smallstep{0}{c(x_1,\dots,x_n)}{\sigma}{c(v_1,\dots,v_n)}{\sigma}}{%
       c \in \CS
       &
       x_1\sigma = v_1 
       & \cdots
       & x_n\sigma = v_n%
     }$
   \\[3ex]      
   $\infer{\smallstep{0}{f(v_1,\dots,v_n)}{\sigma}{f(x_1,\dots,x_n)}{\sigma \dunion \rho}}{%
     \forall i\colon \text{$v_i$ is a value}
     & \rho = \{x_1 \mapsto v_1,\dots,x_n \mapsto v_n\}
     & \text{$f$ is defined and all $x_i$ are fresh}
     }$
   \\[3ex]
   $\infer{\smallstep{1}{f(x_1,\dots,x_n)}{\sigma}{r}{\sigma \dunion \tau}}{%
       f(l_1,\dots,l_n) \to r \in \RS
       & \forall i\colon x_i\sigma = l_i\tau
      }$
   \\[3ex]
   $\infer{\smallstep{1}{f(\dots,t_i,\dots)}{\sigma}{f(\dots,u,\dots)}{\sigma'}}{%
       \smallstep{1}{t_i}{\sigma}{u}{\sigma'}
       }$
  \end{tabular}

\bigskip
\noindent
Note that the substitutions $\sigma$, $\sigma'$, $\tau$, and 
$\rho$ are normalised. 
\caption{Operational Small-Step Semantics}
\label{fig:3}
\end{figure}

The big-step semantics, the type system, and Theorem~\ref{l:2} 
provide an amortised resource analysis for typed TRSs that yields
polynomial bounds. However, Theorem~\ref{l:2} is not
directly applicable, if we want to link this analysis to the interpretation method. 
We recast the method and present a small-step semantics, which will be used in our second soundness
results (Theorem~\ref{l:5} below), cf.~Figure~\ref{fig:3}.
As the big-step semantics, the small-step semantics is decorated
with counters for the derivation height of the evaluated
terms.

Suppose $\smallstep{}{s}{\sigma}{t}{\sigma'}$ holds for terms $s,t$
and substitutions $\sigma,\sigma'$. An inspection of the rules
shows that $\sigma'$ is an extension of $\sigma$. Moreover
we have the following fact.

\begin{lemma}
\label{l:9}
Let $s, t$ be terms, let $\sigma$ be a normalised substitution such that
$\Var(s) \subseteq \dom(\sigma)$ and suppose $\smallstep{}{s}{\sigma}{t}{\sigma'}$. 
Then $\sigma'$ extends $\sigma$ and $s\sigma = s\sigma'$.
\end{lemma}
\begin{proof}
The first assertion follows by induction on the relation
$\smallstep{}{s}{\sigma}{t}{\sigma'}$. Now suppose $\sigma = \rst{\sigma'}{\dom(\sigma)}$.
Then $s\sigma = s(\rst{\sigma'}{\dom(\sigma)}) = s\sigma'$.
\end{proof}

The transitive closure of the judgement $\smallstep{m}{s}{\sigma}{t}{\tau}$
is defined as follows:
\begin{enumerate}
\item $\tsmallstep{m}{s}{\sigma}{t}{\tau}$ if $\smallstep{m}{s}{\sigma}{t}{\tau}$
\item $\tsmallstep{m_1 + m_2}{s}{\sigma}{u}{\rho}$ if $\smallstep{m_1}{s}{\sigma}{t}{\tau}$
and $\tsmallstep{m_2}{t}{\tau}{u}{\rho}$.
\end{enumerate}
The next lemma proves the equivalence of big-step and small-step semantics.

\begin{lemma}
\label{l:3}
Let $\sigma$ be a normalised substitution, let $t$ be a term, 
$\Var(t) \subseteq \dom(\sigma)$, and let $v$ be a value. 
Then $\eval{\sigma}{m}{t}{v}$ if and only if 
$\tsmallstep{m}{t}{\sigma}{v}{\sigma'}$, 
where $\sigma'$ is an extension of $\sigma$.
\end{lemma}
\begin{proof}
First we prove the direction from left to right.
\begin{enumerate}
\item Suppose $\Pi$ has the form: 
\begin{equation*}
    \infer{\eval{\sigma}{0}{x}{v}}{x\sigma = v}
    \tkom
\end{equation*}
such that $t = x$ and $v = x\sigma$. Hence we obtain 
$\tsmallstep{0}{x}{\sigma}{v}{\sigma}$.

\item Suppose $\Pi$ has the form: 
\begin{equation*}
 \infer{\eval{\sigma}{0}{c(x_1,\dots,x_n)}{c(v_1,\dots,v_n)}}{%
       c \in \CS
       &
       x_1\sigma = v_1 
       & \cdots
       & x_n\sigma = v_n%
     }
     \tkom
\end{equation*}
such that $t = c(x_1,\dots,x_n)$ and $v = c(v_1,\dots,v_n)$.
Again, we directly obtain $\tsmallstep{0}{t}{\sigma}{v}{\sigma}$.

\item Suppose the last rule in $\Pi$ if of form:
  \begin{equation*}
    \infer{\eval{\sigma}{m+1}{f(x_1,\dots,x_n)}{v}}{%
      f(l_1,\dots,l_n) \to r \in \RS
       & \forall i\colon x_i\sigma = l_i\tau
       & \eval{\sigma \dunion \tau}{m}{r}{v}
     }
     \tkom
  \end{equation*}
where $t=f(x_1,\dots,x_n)$.
By hypothesis there exists an extension $\sigma'$ of $\sigma \dunion \tau$
such that $\tsmallstep{m}{r}{\sigma \dunion \tau}{v}{\sigma'}$. Furthermore,
we have $\smallstep{1}{t}{\sigma}{r}{\sigma \dunion \tau}$. 
Thus $\tsmallstep{m+1}{t}{\sigma}{v}{\sigma'}$. By definition 
$\dom(\sigma) \cap \dom(\tau) = \varnothing$. Hence 
$\sigma'$ is an extension of $\sigma$.

\item Finally, suppose the last rule in $\Pi$ has the form
  \begin{equation*}
    \infer{\eval{\sigma}{m}{f(t_1,\dots,t_n)}{v}}{%
      \eval{\sigma \dunion \rho}{m_0}{f(x_1,\dots,x_n)}{v}%
      & \eval{\sigma}{m_1}{t_1}{v_1}
      & \cdots
      & \eval{\sigma}{m_n}{t_n}{v_n}
      & m = \sum_{i=0}^n m_i
       }
       \tkom
  \end{equation*}
where $t=f(t_1,\dots,t_n)$.
By induction hypothesis (and repeated use of Lemma~\ref{l:4}), 
we have for all $i=1,\dots,n$: 
$\tsmallstep{m_{i}}{t_1}{\sigma_{i-1}}{v_1}{\sigma_i}$,
where we set $\sigma_0 = \sigma$ and note that all $\sigma_i$ are extensions of $\sigma$. 
As $\smallstep{0}{f(v_1,\dots,v_n)}{\sigma_n}{f(x_1,\dots,x_n)}{\sigma_n \dunion \rho}$
we obtain:
\begin{equation}
  \label{eq:1}
  \tsmallstep{\sum_{i=1}^n m_i}{f(t_1,\dots,t_n)}{\sigma}{f(x_1,\dots,x_n)}{\sigma_n \dunion \rho}
  \tpkt
\end{equation}
Furthermore, by Lemma~\ref{l:4} and the induction hypothesis there exists
a substitution $\sigma'$ such that 
\begin{equation}
  \label{eq:2}
  \tsmallstep{m_0}{f(x_1,\dots,x_n)}{\sigma_n \dunion \rho}{v}{\sigma'}
  \tkom
\end{equation}
where $\sigma'$ extends $\sigma_n \dunion \rho$ (and thus also $\sigma$
as $\dom(\sigma_n) \cap \dom(\rho) = \varnothing$). 
From~\eqref{eq:1} and~\eqref{eq:2} we obtain $\tsmallstep{m}{t}{\sigma}{v}{\sigma'}$.
\end{enumerate}
This establishes the direction from left to right. 
Now we consider the direction form right to left. The proof of the first reduction
$\smallstep{m}{t}{\sigma}{u}{\sigma''}$ in $D$ is denoted as~$\Xi$.
\begin{enumerate}
\item Suppose $\Xi$ has either of the following forms
  \begin{equation*}
    \infer{\smallstep{0}{x}{\sigma}{v}{\sigma}}{%
     x\sigma = v%
     }
     \qquad
     \infer{\smallstep{0}{c(x_1,\dots,x_n)}{\sigma}{c(v_1,\dots,v_n)}{\sigma}}{%
       x_1\sigma = v_1 
       & \cdots
       & x_n\sigma = v_n%
     }
  \end{equation*}
  Then the lemma follows trivially.

\item Suppose $\Xi$ has the form 
  \begin{equation*}
    \infer{\smallstep{0}{f(v_1,\dots,v_n)}{\sigma}{f(x_1,\dots,x_n)}{\sigma \dunion \rho}}{%
     \forall i\colon \text{$v_i$ is a value}
     & \rho = \{x_1 \mapsto v_1,\dots,x_n \mapsto v_n\}
     & \text{$f$ is defined and all $x_i$ are fresh}
     }
     \tpkt
  \end{equation*} 
We apply the induction hypothesis to conclude 
$\eval{\sigma \dunion \rho}{m}{f(x_1,\dots,x_n)}{v}$. Moreover,
we observe that $\eval{\sigma}{0}{v_i}{v_i}$ holds for all $i=1,\dots,n$.
(This follows by a straightforward inductive argument.)
From this we derive $\eval{\sigma}{0}{f(v_1,\dots,v_n)}{v}$ as follows:
\begin{equation*}
  \infer{\eval{\sigma}{0}{f(v_1,\dots,v_n)}{v}}{%
    \eval{\sigma \dunion \rho}{m}{f(x_1,\dots,x_n)}{v}
    & \eval{\sigma}{0}{v_1}{v_1} 
    & \cdots
    & \eval{\sigma}{0}{v_n}{v_n}
    }
    \tpkt
\end{equation*}

\item Suppose $\Xi$ has the form
  \begin{equation*}
    \infer{\smallstep{1}{f(x_1,\dots,x_n)}{\sigma}{r}{\sigma \dunion \tau}}{%
       f(l_1,\dots,l_n) \to r \in \RS
       & \forall i\colon x_i\sigma = l_i\tau
      }
      \tkom
  \end{equation*}
  such that $\sigma'$ is an extension of $\sigma \dunion \tau$.
  By induction hypothesis we conclude $\eval{\sigma \dunion \tau}{m'}{r}{v}$. In conjunction
  with an application of the rule 
  \begin{equation*}
    \infer{\eval{\sigma}{m+1}{f(x_1,\dots,x_n)}{v}}{%
       f(l_1,\dots,l_n) \to r \in \RS
       & \forall i\colon x_i\sigma = l_i\tau
       & \eval{\sigma \dunion \tau}{m'}{r}{v}
     }
     \tkom
  \end{equation*}
  we derive $\eval{\sigma}{m'+1}{f(x_1,\dots,x_n)}{v}$ as required.

\item Suppose $\Xi$ has the form
  \begin{equation*}
    \infer{\smallstep{1}{f(\dots,t_i,\dots)}{\sigma}{f(\dots,u,\dots)}{\sigma''}}{%
       \smallstep{1}{t_i}{\sigma}{u}{\sigma''}
       }
       \tkom
  \end{equation*}
  such that $\sigma'$ is an extension of $\sigma''$. Then by induction hypothesis
  we obtain: $\eval{\sigma''}{m'}{f(\dots,u,\dots)}{v}$. Furthermore by induction hypothesis 
  we have $\eval{\sigma}{1}{t_i}{v_1}$

\item Suppose the initial sequence of $D$ is based on the following reductions,
where $m = \sum_{i=1}^n m_i + m'$.
  \begin{align*}
     & \tsmallstep{m_1}{f(t_1,\dots,t_n)}{\sigma}{f(v_1,\dots,t_n)}{\sigma_1} 
     \\
     & \qquad \vdots
     \\
     & \tsmallstep{m_n}{f(v_1,\dots,t_n)}{\sigma}{f(v_1,\dots,v_n)}{\sigma_n} 
     \\[1ex]
     &\smallstep{0}{f(v_1,t_2,\dots,t_n)}{\sigma_n}{f(x_1,\dots,x_n)}{\sigma_n \dunion \rho}
    \tpkt
  \end{align*}
We apply induction hypothesis on 
$\smallstep{m'}{f(x_1,\dots,x_n)}{\sigma' \dunion \rho}{v}{\sigma'}$ and conclude:
$\eval{\sigma' \dunion \rho}{m'}{f(x_1,\dots,x_n)}{v}$. 
Again by induction hypothesis and inspection of the corresponding proofs,
we obtain $\eval{\sigma_{i-1}}{m_i}{t_i}{v_i}$ for all $i=1,\dots,n$. (We
set $\sigma_0 \defsym \sigma$.)
Due to Lemma~\ref{l:9} we have $t_i\sigma_i = t_i\sigma$. 
Thus, for all $i$, $\eval{\sigma}{m_{i}}{t_i}{v_i}$.
Note that $\dom(\sigma_n) \cap \dom(\rho) = \varnothing$. Hence, 
from $\eval{\sigma_n \dunion \rho}{m'}{f(x_1,\dots,x_n)}{v}$ we obtain
$\eval{\sigma \dunion \rho}{m'}{f(x_1,\dots,x_n)}{v}$. 
Thus $\eval{\sigma}{m}{t}{v}$ follows.
\end{enumerate}
\end{proof}

We extend the notion of potential (cf.~Definition~\ref{d:potential}) to ground
terms.

\begin{definition}
\label{d:potential2}
Let $t = f(t_1,\dots,t_n) \in \TA(\DS \cup \CS)$
and let $\atypdcl{A_1 \times \cdots \times A_n}{C}{p}{q} \in \FS(f)$. Then the
\emph{potential} of $t$ is defined as follows:
\begin{equation*}
  \Phi(\typed{t}{C}) \defsym (p-q) + \Phi(\typed{t_1}{A_1}) + \cdots + \Phi(\typed{t_n}{A_n}) 
  \tpkt
\end{equation*}
\end{definition}

Note that by assumption the declaration in $\FS(f)$ is unique. 

\begin{example}[continued from Example~\ref{ex:5}]
\label{ex:6}
Recall the types of $\queue$ and $\checkF$. Let $q=\queue(f,r)$ be a queue.
We obtain
$\Phi(\typed{\checkF(q)}{\aQueue{(0,1)}}) = 3 + \Phi(\typed{q}{\aQueue{(0,1)}}) 
 = 3 + \Phi(\typed{f}{\aList{0}}) + \Phi(\typed{r}{\aList{1}}) = 3 + \len{r}$.
\end{example}

\begin{lemma}
\label{l:6}
Let $\RS$ and $\sigma$ be well-typed. 
Suppose $\tjudge{\Gamma}{p}{\typed{t}{A}}$. Then we have 
$\Phi(\typed{\sigma}{\Gamma}) + p \geqslant \Phi(\typed{t\sigma}{A})$.
\end{lemma}
\begin{proof}
Let $\Xi$ denote the proof of $\tjudge{\Gamma}{p}{\typed{t}{A}}$. 
\begin{enumerate}
\item Let $t = x$ and thus wlog.\ $\Xi$ is of form
  \begin{equation*}
     \infer{\tjudge{\typed{x}{A}}{0}{\typed{x}{A}}}{}
     \tpkt
  \end{equation*}
Then $\Phi(\typed{\sigma}{\Gamma}) = \Phi(\typed{x\sigma}{A}) = \Phi(\typed{t\sigma}{A})$,
from which the lemma follows.

\item Let $t = f(x_1,\dots,x_n)$ where $f \in \CS \cup \DS$. Thus wlog.\ $\Xi$ is of
form
\begin{equation*}
  \infer{\tjudge{\typed{x_1}{A_1^{\vec{u_1}}},\dots,\typed{x_n}{A_n^{\vec{u_n}}}}{p}{\typed{f(x_1,\dots,x_n)}{\atyp{C}{\vec{v}}}}}{%
      f \in \CS \cup \DS
      &
      \atypdcl{A_1^{\vec{u_1}} \times \cdots \times A_n^{\vec{u_n}}}{\atyp{C}{\vec{v}}}{p} \in \FS(f)
    }
    \tpkt
\end{equation*}
Hence we obtain
\begin{equation*}
  \Phi(\typed{\sigma}{\Gamma}) + p = 
  \sum_{i=1}^n \Phi(\typed{x_i\sigma}{A_i^{\vec{u_i}}}) + p
  = \Phi(\typed{t\sigma}{C^{\vec{v}}})
  \tkom
\end{equation*}
and the lemma follows.

\item Suppose $t=f(t_1,\dots,t_n)$, such that $\vec{t} \not\in \VS$ 
and $f \in \CS \cup \DS$. Thus $\Xi$ is of form
\begin{equation*}
    \infer{\tjudge{\Gamma_1,\dots,\Gamma_n}{p}{\typed{f(t_1,\dots,t_n)}{A}}}{%
      \tjudge{\overbrace{\typed{x_1}{A_1},\dots,\typed{x_n}{A_n}}^{{} =: \Delta}}{p_0}{%
             \typed{f(x_1,\dots,x_n)}{A}}
      & \tjudge{\Gamma_1}{p_1}{\typed{t_1}{A_1}}
      & \cdots
      & \tjudge{\Gamma_n}{p_n}{\typed{t_n}{A_n}}
    }
    \tkom
\end{equation*}
where $p = \sum_{i=0}^n p_i$.
Then by induction hypothesis we have 
$\Phi(\typed{\sigma}{\Gamma_i}) + p_i \geqslant \Phi(\typed{t_i\sigma}{A_i})$ for
all $i=1,\dots,n$.
Hence 
$\sum_{i=1}^n \Phi(\typed{\sigma}{\Gamma_i}) + \sum_{i=1}^n p_i 
\geqslant \sum_{i=1}^n \Phi(\typed{t_i\sigma}{A_i})$.
Let $\rho \defsym \{x_1 \mapsto t_1\sigma,\dots,x_n \mapsto t_n\sigma\}$.
Again by induction hypothesis we have
$\Phi(\typed{\rho}{\Delta}) + p_0 \geqslant \Phi(\typed{f(x_1,\dots,x_n)\rho}{A})$.
Note that $f(x_1,\dots,x_n)\rho = t\sigma$ and $x_i\rho = t_i\sigma$ 
by construction. We obtain
\begin{align*}
  \Phi(\typed{\sigma}{\Gamma}) + \sum_{i=0}^n p_i &= \sum_{i=1}^n \Phi(\typed{\sigma}{\Gamma_i}) + p_0
  \geqslant \sum_{i=1}^n \Phi(\typed{t_i\sigma}{A_i}) + p_0
  \\
  &= \sum_{i=1}^n \Phi(\typed{x_i\rho}{A_i}) + p_0
  = \Phi(\typed{\rho}{\Delta}) + p_0 \\
  & \geqslant \Phi(\typed{t\sigma}{A})
  \tpkt
\end{align*}

\item Suppose $\Xi$ is of form:
  \begin{equation*}
    \infer{\tjudge{\Gamma}{p'}{\typed{t}{C}}}{%
        \tjudge{\Gamma}{p}{\typed{t}{C}}
        &
        p' \geqslant p
        }
  \end{equation*}
By induction hypothesis, we have 
$\Phi(\typed{\sigma}{\Gamma}) + p \geqslant \Phi(\typed{t\sigma}{A})$.
Then the lemma follows from the assumption $p' \geqslant p$. 

\item Suppose $\Xi$ ends with one of the following
structural rules
\begin{equation*}
  \infer{\tjudge{\Gamma, \typed{x}{A}}{p}{\typed{t}{C}}}{%
      \tjudge{\Gamma}{p}{\typed{t}{C}}
      }
  \hspace{10ex}
  \infer{\tjudge{\Gamma, \typed{z}{A}}{p}{\typed{t[z,z]}{C}}}{%
      \tjudge{\Gamma, \typed{x}{A_1}, \typed{y}{A_2}}{p}{\typed{t[x,y]}{C}}%
      & 
      \share{A}{A_1,A_2}
      }
\end{equation*}
We only consider the second rule, as the first alternatives
follows trivially. 
Let $\rho \defsym \sigma \dunion \{x \mapsto z\sigma, y \mapsto z\sigma\}$;
by induction hypothesis, we have 
$\Phi(\typed{\rho}{\Gamma,\typed{x}{A_1}, \typed{y}{A_2} }) + p \geqslant 
\Phi(\typed{t[x,y]\rho}{A})$. 
By definition of $\rho$ and Lemma~\ref{l:10}, we obtain
\begin{equation*}
  \Phi(\typed{\sigma}{\Gamma,\typed{z}{A}}) = 
  \Phi(\typed{\rho}{\Gamma, \typed{x}{A_1}, \typed{y}{A_2}})
  \tpkt
\end{equation*}
Hence $\Phi(\typed{\sigma}{\Gamma,\typed{z}{A}}) + p \geqslant \Phi(\typed{t[z,z]\sigma}{A})$
follows from $t[x,y]\rho = t[z,z]\sigma$.

\item Suppose $\Xi$ ends either in a sub- or in a supertyping rule:
\begin{equation*}
\infer{\tjudge{\Gamma, \typed{x}{A}}{p}{\typed{t}{C}}}{%
      \tjudge{\Gamma, \typed{x}{B}}{p}{\typed{t}{C}}
      &
      A \subtype B
    }
\hspace{10ex}
    \infer{\tjudge{\Gamma}{p}{\typed{t}{C}}}{%
      \tjudge{\Gamma}{p}{\typed{t}{D}}
      &
      D \subtype C
      }
\end{equation*}
Consider the second rule. We have to show that 
$\Phi(\typed{\sigma}{\Gamma}) + p \geqslant \Phi(\typed{t\sigma}{C})$.
This follows from induction hypothesis, which yields
$\Phi(\typed{\sigma}{\Gamma}) + p \geqslant \Phi(\typed{t\sigma}{D})$
as $\Phi(\typed{t\sigma}{D}) \geqslant  \Phi(\typed{t\sigma}{C})$ by
definition of the subtyping relation. 
The argument for the first rule is similar.
This concludes the inductive argument.
\end{enumerate}
\end{proof}

We obtain our second soundness result.

\begin{theorem}
\label{l:5}
Let $\RS$ and $\sigma$ be well-typed.
Suppose $\tjudge{\Gamma}{p}{\typed{t}{A}}$ and 
$\smallstep{m}{t}{\sigma}{u}{\sigma'}$.
Then $\Phi(\typed{\sigma}{\Gamma}) - \Phi(\typed{u\sigma'}{A}) + p \geqslant m$.
Thus if for all ground basic terms $t$ and types $A$: 
$\Phi(\typed{t}{A}) \in \bigO(n^k)$, where $n = \size{t}$, 
then $\rc_{\RS}(n) \in \bigO(n^k)$.
\end{theorem}
\begin{proof}
Let $\Pi$ be the proof of the judgement $\smallstep{m}{t}{\sigma}{u}{\sigma'}$
and let $\Xi$ denote the proof of $\tjudge{\Gamma}{p}{\typed{t}{A}}$. 
The proof proceeds by main-induction on the length of $\Pi$
and by side-induction on the length of $\Xi$.
We focus on some interesting cases.

\begin{enumerate}
\item Suppose $\Pi$ has the form 
\begin{equation*}
    \infer{\smallstep{0}{x}{\sigma}{u}{\sigma}}{%
     x\sigma = u%
     }
    \tkom
\end{equation*}
such that $t = x$ and $u = x\sigma$. As $\sigma$ is normalised $u$ is
a value. Wlog.\ we can assume that $\Xi$ is of form $\tjudge{\typed{x}{A}}{0}{\typed{x}{A}}$. 
It suffices to show $\Phi(\typed{\sigma}{\Gamma}) \geqslant \Phi(\typed{u\sigma}{A})$,
which follows from Lemma~\ref{l:6} as $x\sigma = u = u\sigma$.

\item Suppose $\Pi$ has the form
\begin{equation*}
    \infer{\smallstep{0}{c(x_1,\dots,x_n)}{\sigma}{c(u_1,\dots,u_n)}{\sigma}}{%
       x_1\sigma = u_1 
       & \cdots
       & x_n\sigma = u_n%
     }
     \tkom
\end{equation*}
such that $t = c(x_1,\dots,x_n)$ and $u = c(x_1\sigma,\dots,x_n\sigma)$, which
again is a value. Further let $\Xi$ end in the judgement:
\begin{equation*}
    \tjudge{\typed{x_1}{A_1^{\vec{u_1}}},\dots,\typed{x_n}{A_n^{\vec{u_n}}}}{p}{\typed{c(x_1,\dots,x_n)}{\atyp{C}{\vec{v}}}}
    \tpkt
\end{equation*}
Let $\Gamma = \typed{x_1}{A_1^{\vec{u_1}}},\dots,\typed{x_n}{A_n^{\vec{u_n}}}$; by Lemma~\ref{l:6} we have 
$\Phi(\typed{\sigma}{\Gamma}) + p \geqslant \Phi(\typed{t\sigma}{A}) = 
\Phi(\typed{u\sigma}{A})$ as $t\sigma = u = u\sigma$.

\item Suppose $\Pi$ has the form 
  \begin{equation*}
   \infer{\smallstep{0}{f(v_1,\dots,v_n)}{\sigma}{f(x_1,\dots,x_n)}{\sigma \dunion \rho}}{%
     \forall i\colon \text{$v_i$ is a value}
     & \rho = \{x_1 \mapsto v_1,\dots,x_n \mapsto v_n\}
     & \text{$f$ is defined and all $x_i$ are fresh}
     }    
  \end{equation*}
Then $t = f(v_1,\dots,v_n)$ is ground, as all $v_i$ are values. Hence, we have
\begin{equation*}
  t\sigma = t =  f(x_1,\dots,x_n)\rho = f(x_1,\dots,x_n)(\sigma \dunion \rho)
  \tpkt
\end{equation*}
The last equality follows as $\dom(\sigma) \cap \dom(\rho) = \varnothing$.
By Lemma~\ref{l:6} we have $\Phi(\typed{\sigma}{\Gamma}) + p \geqslant \Phi(\typed{t\sigma}{A})$.
Then the theorem follows as $t\sigma = f(x_1,\dots,x_n)(\sigma \dunion \rho)$ from
above.
 
\item Suppose $\Pi$  has the form
  \begin{equation*}
  \infer{\smallstep{1}{f(x_1,\dots,x_n)}{\sigma}{r}{\sigma \dunion \tau}}{%
       f(l_1,\dots,l_n) \to r \in \RS
       & \forall i\colon x_i\sigma = l_i\tau
      }
     \tpkt
  \end{equation*}
Then $t=f(x_1,\dots,x_n)$ and $f(x_1,\dots,x_n)\sigma = f(l_1,\dots,l_n)\tau$.
Suppose $\Var(f(\vec{l})) = \{y_{1},\dots,y_{\ell}\}$ and let
$\Var(l_i) = \{y_{i1},\dots,y_{il_i}\}$ for $i \in \{1,\dots,n\}$. 
As $\RS$ is left-linear we have $\Var(f(l_1,\dots,l_n)) = \biguplus_{i=1}^n \Var(l_i)$.
We set  $\Gamma = \typed{x_1}{A_1},\dots,\typed{x_n}{A_n}$.
By the assumption $\tjudge{\Gamma}{p}{\typed{t}{A}}$ and 
well-typedness of $\RS$ we obtain 
\begin{equation}
\label{eq:5}
\tjudge{\overbrace{\typed{y_1}{B_1}, \dots, \typed{y_{\ell}}{B_\ell}}^{{} =: \Delta}}{p - 1 + \sum_{i=1}^n k_i}{\typed{r}{C}}
    \tkom
\end{equation}
as in~\eqref{eq:welltyped}. We have
\begin{align*}
  \Phi(\typed{\sigma}{\Gamma}) + p & = \sum_{i=1}^n \Phi(\typed{x_i\sigma}{A_i}) + p\\
  & = \sum_{i=1}^n \left( k_i + \Phi(\typed{y_{i1}\tau}{B_{i1}}) + \cdots + 
                        \Phi(\typed{y_{il_i}\tau}{B_{il_i}}) \right) + p \\
  & = \Phi(\typed{\tau}{\Delta}) + \sum_{i=1}^n k_i + (p-1) + 1 
  \\
  & \geqslant \Phi(\typed{r\tau}{C}) + 1 
  \geqslant \Phi(\typed{r(\sigma \dunion \tau)}{C}) +1
  \tpkt
\end{align*} 
Here the first equality follows by an inspection on the cases for the constructors and
$\Phi(\typed{\tau}{\Delta}) + \sum_{i=1}^n k_i + (p-1) \geqslant \Phi(\typed{r\tau}{C})$
follows due to Lemma~\ref{l:6} and~\eqref{eq:5}. Furthermore note
that $r\tau = r(\sigma \dunion \tau)$, as $\dom(\sigma) \cap \dom(\tau) = \varnothing$.

\item Suppose the last rule in $\Pi$ has the form
\begin{equation*}
   \infer{\smallstep{1}{f(t_1,\dots,t_n)}{\sigma}{f(u,\dots,t_n)}{\sigma'}}{%
       \smallstep{1}{t_1}{\sigma}{u}{\sigma'}
       }
       \tpkt
\end{equation*}
Wlog.\ the last rule in the type inference $\Xi$ is of the following form,
where we can assume that every variable occurs at most once in $f(t_1,\dots,t_n)$.
\begin{equation*}
  \infer{\tjudge{\underbrace{\Gamma_1,\dots,\Gamma_n}_{{} =: \Gamma}}{p}{\typed{f(t_1,\dots,t_n)}{C}}}{%
      & \tjudge{\overbrace{\typed{x_1}{A_1},\dots,\typed{x_n}{A_n}}^{{} =: \Delta}}{p_0}{%
             \typed{f(\vec{x})}{C}}
      \tjudge{\Gamma_1}{p_1}{\typed{t_1}{A_1}}
      & \cdots
      & \tjudge{\Gamma_n}{p_n}{\typed{t_n}{A_n}}
      & p = \sum_{i=0}^n p_i%
    }
        \tpkt
\end{equation*}
By induction hypothesis on $\smallstep{1}{t_1}{\sigma}{u}{\sigma'}$ 
and $\tjudge{\Gamma_1}{p_1}{\typed{t_1}{A_1}}$ we obtain 
(i) $\Phi(\typed{\sigma}{\Gamma_1}) - \Phi(\typed{u\sigma'}{A_1}) + p_1 \geqslant
1$ and $n-1$ applications of Lemma~\ref{l:6} yield
(ii) $\Phi(\typed{\sigma}{\Gamma_i}) + p_{i} \geqslant \Phi(\typed{t_i\sigma}{A_i})$
for all $i=2,\dots,n$. 
We set $\rho \defsym \{x_1 \to u\sigma', x_2 \to t_2\sigma, \dots, x_n \to t_n\sigma\}$.
Another application of Lemma~\ref{l:6} on $\tjudge{\Delta}{p_0}{\typed{f(x_1,\dots,x_n)}{C}}$ yields 
(iii) $\Phi(\typed{\rho}{\Delta}) + p_0 \geqslant \Phi(\typed{f(x_1\rho,x_2\rho,\dots,x_n\rho)}{C})$. 
Finally, we observe $\Phi(\typed{\sigma}{\Gamma}) = \sum_{i=1}^n \Phi(\typed{\sigma}{\Gamma_i}$.
The theorem follows by combining the equations in (i)--(iii).

\item Suppose $\Xi$ is of form:
\begin{equation*}
  \infer{\tjudge{\Gamma}{p'}{\typed{t}{C}}}{%
    \tjudge{\Gamma}{p}{\typed{t}{C}}
    &
    p' \geqslant p
  }
\end{equation*}
By side-induction on $\tjudge{\Gamma}{p}{\typed{t}{C}}$
and $\smallstep{m}{t}{\sigma}{u}{\sigma'}$ we conclude
that $\Phi(\typed{\sigma}{\Gamma}) - \Phi(\typed{u\sigma'}{A}) + p \geqslant m$.
Then the theorem follows from the assumption $p' \geqslant p$. 

\item Suppose $\Xi$ is of form:
\begin{equation*}
  \infer{\tjudge{\Gamma, \typed{x}{A}}{p}{\typed{t}{C}}}{%
      \tjudge{\Gamma}{p}{\typed{t}{C}}
      }
\end{equation*}
We conclude by side-induction that 
$\Phi(\typed{\sigma}{\Gamma}) - \Phi(\typed{u\sigma'}{A} + p \geqslant m$.
As $\Phi(\typed{\sigma}{\Gamma, \typed{x}{A}}) \geqslant \Phi(\typed{\sigma}{\Gamma})$
the theorem follows.

\item Suppose $\Xi$ is of form:
\begin{equation*}
  \infer{\tjudge{\Gamma, \typed{z}{A}}{p}{\typed{t[z,z]}{C}}}{%
      \tjudge{\Gamma, \typed{x}{A_1}, \typed{y}{A_2}}{p}{\typed{t[x,y]}{C}}%
      & 
      \share{A}{A_1,A_2}
      }
\end{equation*}
By assumption $\smallstep{m}{t[z,z]}{\sigma}{u}{\sigma'}$; let
$\rho \defsym \sigma \dunion \{x \mapsto z\sigma, y \mapsto z\sigma\}$.
By side-induction on 
$\tjudge{\Gamma, \typed{x}{A_1}, \typed{y}{A_2}}{p}{\typed{t[x,y]}{C}}$
and  $\smallstep{m}{t[x,y]}{\rho}{u}{\sigma'}$ we conclude that for 
all $\Phi(\typed{\rho}{\Gamma,\typed{x}{A_1}, \typed{y}{A_2}}) - 
\Phi(\typed{u\sigma'}{A}) + p \geqslant m$. 
By definition of $\rho$ and Lemma~\ref{l:10}, we obtain
$\Phi(\typed{\sigma}{\Gamma,\typed{z}{A}}) = 
\Phi(\typed{\rho}{\Gamma, \typed{x}{A_1}, \typed{y}{A_2}})$, from which the theorem follows.

\item Suppose $\Xi$ ends either in a sub- or in a supertyping rule:
\begin{equation*}
\infer{\tjudge{\Gamma, \typed{x}{A}}{p}{\typed{t}{C}}}{%
      \tjudge{\Gamma, \typed{x}{B}}{p}{\typed{t}{C}}
      &
      A \subtype B
    }
\hspace{10ex}
    \infer{\tjudge{\Gamma}{p}{\typed{t}{C}}}{%
      \tjudge{\Gamma}{p}{\typed{t}{D}}
      &
      D \subtype C
      }
\end{equation*}
Consider the first rule.
By assumption $\smallstep{m}{t}{\sigma}{u}{\sigma'}$ and by
definition $\Phi(\typed{\sigma}{\Gamma, \typed{x}{A}}) \geqslant 
\Phi(\typed{\sigma}{\Gamma, \typed{x}{B}})$. Thus the theorem follows
by side-induction hypothesis.
\end{enumerate}  
\end{proof}

\section{Typed Polynomial Interpretations}
\label{TypedPI}

We adapt the concept of polynomial interpretation to typed
TRSs. For that we suppose a mapping $\interdomain{\cdot}$ that assigns to
every \emph{annotated} type $C$ a subset of the natural numbers, whose elements are
ordered with $>$ in the standard way. The set $\interdomain{C}$ is called the
\emph{interpretation} of $C$.

\begin{definition}
\label{d:interpretation}
An \emph{interpretation $\gamma$ of function symbols} is a mapping from
function symbols and types to functions over $\N$.
Consider a function symbol $f$ and an annotated type $C$ such that
$\FS(f) \owns \atypdcl{A_1 \times \cdots \times A_n}{C}{p}$. Then the interpretation 
$\gamma(f,C) \colon \interdomain{A_1} \times \cdots \times \interdomain{A_n} \to 
\interdomain{C}$ of $f$ is defined as follows:
\begin{equation*}
  \gamma(f,C)(x_1,\dots,x_n) \defsym x_1 + \cdots + x_n + p \tpkt
\end{equation*}
\end{definition}

Note that by assumption the declaration in $\FS(f)$ is unique and thus
$\gamma(f,C)$ is unique. Interpretations of function symbols naturally extend to
interpretation on ground terms.
\begin{equation*}
  \groundinter{\typed{f(t_1,\dots,t_n)}{C}} \defsym 
  \gamma(f,C)(\groundinter{\typed{t_1}{A_1}},\dots,\groundinter{\typed{t_n}{A_n}})
    \tpkt
\end{equation*}
Let $\RS$ be a well-typed and let the interpretation $\gamma$ of 
function symbols in $\FS$ be induced by the well-typing of $\RS$. 
Then by construction $\groundinter{\typed{t}{A}} = \Phi(\typed{t}{A})$.

\begin{example}[continued from Example~\ref{ex:5}]
\label{ex:7}
Based on Definition~\ref{d:interpretation} we obtain the following definitions
of the interpretation of function symbols $\gamma$. We start
with the constructor symbols.
\begin{alignat*}{4}
  \gamma(\zero,\aNat{p}) &= 0 & \hspace{2ex}
  \gamma(\mS,\aNat{p})(x) &= x + p & \hspace{2ex}
  \gamma(\errorHead,\aNat{p}) &= 0 
  \\
  \gamma(\nil,\aList{q}) &= 0 & \hspace{2ex}
  \gamma(\cons,\aList{q})(x,y) &= x + y + q & \hspace{2ex}
  \gamma(\errorTail,\aQueue{(0,1)}) &= 0 
  \\
  \gamma(\queue,\aQueue{(0,1)})(x,y) &= x + y \tkom 
\end{alignat*}
where $p,q \in \N$. 
Similarly the definition of $\gamma$ for defined symbols follows from the
signature detailed in Example~\ref{ex:5}. 
It is not difficult to see that for any rule $l \to r \in \RSa$ and any
substitution $\sigma$, we obtain $\groundinter{l\sigma} > \groundinter{r\sigma}$.
We show this for rule~1. 
\begin{align*}
  \groundinter{\typed{\checkF(\queue(\nil,r\sigma))}{\aQueue{(0,1)}}} &= 
  \groundinter{\typed{r\sigma}{\aList{1}}} + 3 > 0 \\
  & = \groundinter{\typed{\rev(r\sigma)}{\aList{0}}} + \groundinter{\typed{\nil}{\aList{1}}} \\
  & = \groundinter{\typed{\queue(\rev(r\sigma),\nil)}{\aQueue{(0,1)}}} \tpkt
\end{align*}
%
% Consider the term $\typed{\enq(\mS^n(\zero))}{\aQueue{(0,1)}}$. We obtain for $n \in \N$:
% %
% \begin{align*}
%   \groundinter{\typed{\enq(\mS^{n+1}(\zero))}{\aQueue{(0,1)}}} &= 1 + 6 \cdot (n+1) \\
%   &> 6\cdot n + 6 \\
%   &= \groundinter{\typed{\snoc(\enq(\mS^{n}(\zero)),\mS^{n}(\zero))}{\aQueue{(0,1)}}} \tpkt
% \end{align*}
Orientability of $\RSa$ with the above given interpretation implies
the optimal linear innermost runtime complexity.
\end{example}

We lift the standard order $>$ on the interpretation domain $\N$ to an order on terms
as follows. Let $s$ and $t$ be terms of type $A$. Then $s > t$ if 
for all well-typed substitutions $\sigma$ we have
$\groundinter{\typed{s\sigma}{A}} > \groundinter{\typed{t\sigma}{A}}$.

\begin{theorem}
\label{l:8}
Let $\RS$ be well-typed, constructor TRS over signature $\FS$ and 
let the interpretation of function symbols $\gamma$ 
be induced by the type system. Then $l > r$ 
for any rule ${l \to r} \in \RS$.
Thus if for all ground basic terms $t$ and types $A$: $\groundinter{\typed{t}{A}} \in \bigO(n^k)$, 
where $n = \size{t}$, then $\rc_{\RS}(n) \in \bigO(n^k)$.
\end{theorem}
\begin{proof}
Let $l = f(l_1,\dots,l_n)$ and let $x_1,\dots,x_n$ be
fresh variables. Suppose further $\FS(f) \owns \atypdcl{A_1 \times \cdots \times A_n}{C}{p}$. 
As $\RS$ is well-typed we have 
\begin{equation*}
  \tjudge{\overbrace{\typed{x_1}{A_1},\dots,\typed{x_n}{A_n}}^{{} =: \Gamma}}{p}{\typed{f(x_1,\dots,x_n)}{C}}
  \tkom
\end{equation*}
for $p \in \N$. 

Now suppose that $\tau$ denotes any well-typed substitution for the rule $l \to r$. 
It is standard way, we extend $\tau$ to a well-typed substitution $\sigma$ 
such that $l\tau = f(x_1,\dots,x_n)\sigma$.
By definition of the small-step semantics, we obtain
\begin{equation*}
  \smallstep{1}{f(x_1,\dots,x_n)}{\sigma}{r}{\sigma \dunion \tau} \tpkt
\end{equation*}
Then by Lemma~\ref{l:5}, $\Phi(\typed{\sigma}{\Gamma}) + p > \Phi(\typed{r(\sigma \dunion \tau)}{C})$
and by definitions, we have:
\begin{equation*}
  \Phi(\typed{l\tau}{C}) = \Phi(\typed{f(x_1\sigma,\dots,x_n\sigma)}{C})
  = \sum_{i=1}^n \Phi(\typed{x_i\sigma}{A_i}) + p
  = \Phi(\typed{\sigma}{\Gamma}) + p 
  \tpkt
\end{equation*}
Furthermore, observe that $r(\sigma \dunion \tau) = r\tau$ as 
$\dom(\sigma) \cap \dom(\tau) = \varnothing$. In sum, we
obtain $\Phi(\typed{l\tau}{C}) > \Phi(\typed{r\tau}{C})$, from which we conclude
$\groundinter{\typed{l\tau}{C}}{\gamma} > \groundinter{\typed{r\tau}{C}}$. 
As $\tau$ was chosen arbitrarily, we obtain ${\RS} \subseteq {>}$.
\end{proof}

We say that an interpretation \emph{orients} a typed TRS $\RS$, if
${\RS} \subseteq {>}$. As an immediate consequence of the theorem, we obtain the following 
corollary.
\begin{corollary}
\label{c:1}
Let $\RS$ be a well-typed and constructor TRS. Then there exists a typed
polynomial interpretation over $\N$ that orients $\RS$. 
\end{corollary}

At the end of Section~\ref{AnnotatedTypes} we have remarked on the
automatabilty of the obtained amortised analysis. Observe that
Theorem~\ref{l:8} gives rise to a conceptually quite different
implementation. Instead of encoding the constraints of the
typing rules in Figure~\ref{fig:2} one directly encode
the orientability constraints for each rule, cf.~\cite{contejean:2005}. 

\section{Conclusion}
\label{Conclusion}

This paper is concerned with the connection between amortised resource
analysis, originally introduced for functional programs, and polynomial
interpretations, which are frequently used in complexity and termination
analysis of rewrite systems. 

In order to study this connection we established
a novel resource analysis for typed term rewrite systems based on a potential-based
type system. This type system gives rise to polynomial bounds for innermost runtime complexity.
A key observation is that the classical notion of potential can be altered
so that not only values but any term can be assigned a potential. Ie.\ the potential function
$\Phi$ is conceivable as an interpretation. Based on this observation we have shown 
that well-typedness of a TRSs $\RS$ induces a typed polynomial
interpretation orienting $\RS$.

Apart from clarifying the connection between amortised resource analysis
and polynomial interpretation our results seems to induce two new
methods for the innermost runtime complexity of typed TRSs as indicated above.

We emphasise that these methods are not restricted to typed TRSs, 
as our cost model gives rise to a \emph{persistent} property. Here a property is persistent if, 
for any typed TRS $\RS$ the property holds iff it holds for the corresponding untyped
TRS $\RS'$. While termination is in general not persistent~\cite{TeReSe}, it is
not difficult to see that the runtime complexity is a persistent property. This
is due to the restricted set of starting terms. Thus it seems that the proposed techniques
directly give rise to novel methods of automated innermost runtime complexity analysis.

In future work we will clarify whether the established results extend to
the multivariate amortised resource analysis presented in~\cite{HAH12b}. Furthermore,
we will strive for automation to assess the viability of the established
methods.

%\bibliographystyle{abbrv} 
%\bibliography{references}

\end{document}